\documentclass[journal]{IEEEtran}
\usepackage{ifpdf}
\usepackage{cite}
\usepackage[pdftex]{graphicx}
\usepackage{amsmath}
\usepackage{algorithmic}
\usepackage{array}
\usepackage{caption}
\usepackage[caption=false,font=footnotesize]{subfig}
\usepackage{fixltx2e}
\usepackage{stfloats}
\usepackage{url}
\usepackage{amsthm}
\usepackage{amsmath}
\usepackage{amsfonts}
\usepackage{amssymb}
\usepackage[T1]{fontenc} % optional
\usepackage{amsmath}
\usepackage{txfonts}
\usepackage{bm} % optional
\usepackage[all]{xy}
\usepackage{enumitem}
\usepackage{float}
\usepackage{amsmath}
\usepackage[usenames, dvipsnames]{color}

\theoremstyle{definition}

\newtheorem{thm}{Theorem}

\newtheorem{lem}{Lemma}

\newtheorem{define}{Definition}

\newcommand{\var}{\textrm{Var}}

\newcommand{\no}{\nonumber}

\begin{document}

\title{Achieving Perfect Location Privacy in Wireless Devices Using Anonymization}

\author{Zarrin~Montazeri~\IEEEmembership{Student Member,~IEEE,}
        Amir~Houmansadr~\IEEEmembership{Member,~IEEE,}
        Hossein~Pishro-Nik~\IEEEmembership{Member,~IEEE,}
        %\thanks{Manuscript received April 19, 2005; revised August 26, 2015.}
        \thanks{Z. Montazeri is with the Department
        of Electrical and Computer Engineering, University of Massachusetts, Amherst,
        MA, 01003 USA e-mail: (seyedehzarin@umass.edu).}% <-this % stops a space
        \thanks{A. Houmansadr is with the College of Information and Computer Sciences, University of Massachusetts, Amherst,
        MA, 01003 USA e-mail:(amir@cs.umass.edu)}
        \thanks{H. Pishro-Nik is with the Department
        of Electrical and Computer Engineering, University of Massachusetts, Amherst,
        MA, 01003 USA e-mail:(pishro@engin.umass.edu)}
        \thanks{This work was supported by National Science Foundation under grants CCF 0844725 and CCF 1421957. Parts of this work was presented in Annual Conference on Information Science and Systems (CISS 2016)~\cite{montazeri2016defining}, and in International Symposium on Information Theory and Its Applications (ISITA 2016)~\cite{Mont1610Achieving}.}}

%\markboth{IEEE TRANSACTIONS ON WIRELESS COMMUNICATIONS}
%\IEEEpubid{0000--0000/00\$00.00~\copyright~2015 IEEE}
\maketitle

\begin{abstract}

The popularity of mobile devices and location-based services (LBS) has created great concern regarding the location privacy of their users. Anonymization is a common technique that is often used to protect the location privacy of LBS users. Here, we present an information-theoretic approach to define the notion of perfect location privacy. We show how LBS's should use the anonymization method to ensure that their users can achieve perfect location privacy.

First, we assume that a user's current location is independent from her past locations. Using this i.i.d model, we show that if the pseudonym of the user is changed before $O(n^{\frac{2}{r-1}})$ observations are made by the adversary for that user, then the user has perfect location privacy. Here, $n$ is the number of the users in the network and $r$ is the number of all possible locations that users can go to.

Next, we model users' movements using Markov chains to better model real-world movement patterns. We show that perfect location privacy is achievable for a user if the user's pseudonym is changed before $O(n^{\frac{2}{|E|-r}})$ observations are collected by the adversary for the user, where $|E|$ is the number of edges in the user's Markov chain model.

\end{abstract}


%\IEEEpeerreviewmaketitle

\section{Introduction}

\IEEEPARstart{M}{obile} devices offer a wide range of services by recording and processing the geographic locations of their users. Such services, broadly known as location-based services (LBS), include navigation, ride-sharing, dining recommendation, and auto-collision warning applications. While such LBS applications offer a wide range of popular and important services to their users, they impose significant privacy threats because of their access to the location information of these wireless devices. Privacy compromises can also be launched by other types of adversaries including third-party applications, nearby mobile users, and cellular service providers.

%To protect the location privacy of users of LBS systems, location privacy protection mechanisms (LPPM) have been proposed. These mechanisms tend to perturb the mobile user's location information or the user's identity information before it is disclosed to an LBS system's operators and/or other users.

To protect the location privacy of LBS users, various mechanisms have been designed~\cite{shokri2012protecting, gruteser2003anonymous, hoh2007preserving}, which are known as Location-Privacy Protection Mechanisms (LPPM). These mechanisms tend to perturb the information of wireless devices, such as the user's identifier or location coordinates, before they get disclosed to the LBS application. LPPMs can be classified into two classes; those that perturb the user's identity information are known as \textit{identity perturbation mechanisms} while those that perturb the location information of the users are known as \textit{location perturbation mechanisms}.
Improving the location privacy of users using these LPPMs usually comes at the price of performance degradation for the underlying LBS's, so finding the optimal LPPM with respect to the LBS is still a great concern.

In the proposed framework, we employ the anonymization technique to hide the identity of users over time. First, we assume that each user's current location is independent from her past locations in order to simplify the derivations. Then, we model the users' movements by Markov chains which is a more realistic setting by considering the dependencies between locations over time. When dealing with privacy, it is advisable to assume the strongest model for the adversary, so in this paper we assume that the adversary has complete statistical knowledge of the users' movements. We formulate a user's location privacy based on the mutual information between the adversary's anonymized observations and the user's actual location information. We define the notion of \emph{perfect location privacy} and show that with a properly designed anonymization method, users can achieve perfect location privacy. %Due to space limitations, when proving the results, we refer to ~\cite{Mont1603:Defining} when applicable, and only focus on the parts that are novel and different from those we already proved in ~\cite{Mont1603:Defining}.

Parts of this work has been previously presented as two conference
publications~\cite{montazeri2016defining,Mont1610Achieving}.
In this manuscript, we extend our previous work in \cite{montazeri2016defining,Mont1610Achieving}
by offering new results, analysis, and perspective.
In particular, we provide a clean-slate proof for Theorem~\ref{two_state_thm}
to make sure all parts of the proof are presented in a rigorous way.
More specifically, the old proof in \cite{montazeri2016defining} was only presented in a summary format; a few lemmas were stated without proof or with just a sketch of a proof.
On the other hand, here we provide the new proof along with all the required details.
%\amir{This is not clear to me at all. What is the advantage of the new proof over the older proof?
%what does more rigorous means? was the previous proof faulty?}
Also, we have revised and extended our discussions. For example, Section \ref{sec:example} is added to the paper to clearly elaborate on the problem setting in the paper.

It is worth noting that this paper focuses on the theoretical foundations of the location privacy problem when anonymization is used as an LPPM mechanism. Needless to say,
when implementing anonymization-based LPPMs in practice our analysis needs
to get adjusted to each scenario's specific threat model, e.g., the number of mobile entities,
the capabilities of the adversary (the number and location of observation points, prior knowledge about mobile entities, etc.), the extent of possible geographic locations, etc.

%
%\amir{could you make it more specific by
%stating what discussions we have added? For instance, do we offer new example scenarios? Do we
%perform simulations for larger dataset?}

%%%%%%%%%%%%%%%%%%%%%%%%%%%%%%%%%%%%%%%%%%%%%
\section{Related Work}
%Working on location privacy field has been
Location privacy has been an active field of research over the past decade~\cite{wernke2014classification, wang2015toward, wang2015privacy, niu2015enhancing, cai2015cloaking,kido2005anonymous, lu2008pad, shokri2011quantifying, shokri2011quantifying2, shokri2012protecting}. Studies in this field
can be classified into two main classes: 1) Designing effective LPPMs for specific LBS systems and platforms, 2) Deriving theoretical models for location privacy, e.g., deriving metrics to quantify the location privacy.

The designed LPPMs can be classified into two classes: 1) Location perturbation LPPMs, 2) Identity perturbation LPPMs.
Location perturbation LPPMs aim at obfuscating the location
information of the users over time and geographical domain with methods such as cloaking,~\cite{wernke2014classification, cai2015cloaking}, and adding dummy locations,~\cite{kido2005anonymous, lu2008pad}.
On the other hand, identity perturbation LPPMs try to obfuscate the user's identity while using an LBS. Some common approaches to perturb the identity of the user is to either exchange users' identifiers,~\cite{beresford2004mix}, or assign random pseudonyms to them, known as anonymization technique,~\cite{sweeney2002k, ghinita2007prive}. The former method usually uses some pre-defined regions, called \textit{mixed-zones}, to exchange users' identifiers within those regions. As users cross such regions, they exchange their identifiers with other users in the same region using an encryption protocol to confuse the adversary,~\cite{beresford2003location, freudiger2009optimal}. 

Previous studies have shown that using anonymization alone is not enough to protect users' location privacy in real-world scenarios where users go to unique locations. Particularly, Zang et al.\ demonstrate that an adversary has a significant advantage in identifying users who visit unique locations \cite{zang2011anonymization}. Also, Golle and Partridge show that the possibility of user identification based on anonymized location traces is significantly increased when the individual's home and work locations are known~\cite{golle2009anonymity}.
Please note that this does not contradict the analysis and findings of our paper as we use a different setting. 
First, our analysis seeks to find the theoretical limits of privacy for situations where the number of users ($N$) goes to infinity, which is not the case in previous studies like~\cite{golle2009anonymity,zang2011anonymization}.
Increasing the number of user reduces an adversary's confidence in distinguishing different users. 
Second, in our analysis we assume 
``continuous'' density functions for the movements of the users across different locations (e.g., the $f_P(p)$ function in Section~\ref{sec:model}). Therefore, user distributions do not contain Dirac delta functions representing their unique  locations. 
Note that this is not an unrealistic assumption; 
in real-world scenarios with users having unique locations, we assume that the location information is pre-processed to satisfy this continuity requirement. Such pre-processing can be performed in two ways; first, by reducing the granularity of locations, e.g., in our analysis we divide a region of interest into a number of girds (i.e., into $r$ coarse-grained locations). 
Second, an obfuscation mechanism can be applied to location traces to ensure they satisfy the continuity requirement. 
Further discussion on implementing such pre-processing is out of the scope of our work and we leave it to future work.

A related, but in parallel, approach to our study is differential privacy-based mechanisms.
Differential privacy is mainly studied in the context of databases containing sensitive information, where the goal of differential privacy is to respond queries on the aggregation of the information in the database without revealing sensitive information about the individual entries in the database. 
Differential privacy has been extensively studied in the context of location privacy, i.e., to prevent data leakage from location information databases~\cite{wang2015privacy, lee2012differential, bordenabe2014optimal, chatzikokolakis2015geo, nguyen2013differential, machanavajjhala2008privacy}. 
The goal here is to insure that the presence of no single user could noticeably change the outcome of the aggregated location information. For instance, Ho et al.~\cite{ho2011differential} proposed a differentially private location pattern mining algorithm using quadtree spatial decomposition.
Some location perturbation LPPMs are based on ideas from differential privacy~\cite{chatzikokolakis2013broadening, shokri2014optimal, chatzikokolakis2015location, andres2013geo, bordenabe2014optimal}. 
For instance, Dewri~\cite{dewri2013local} suggest to design obfuscation LPPMs by applying differential perturbations.
Alternatively, Andres et al.\ hide the exact location of each user in a region by adding Laplacian distributed noise  to achieve a desired level of geo-indistinguishability~\cite{andres2013geo}.
Note that our approach is entirely in parallel with this line of work. 
Our paper tries to achieve the theoretical limits on location privacy\textemdash independent of the LPPM mechanisms being used\textemdash while differential privacy based studies on location privacy try to \emph{design} specific LPPM mechanisms under very specific application scenarios.

Several works aim at quantifying the location privacy of mobile users. A common approach is called K-anonymity,~\cite{gruteser2003anonymous, sweeney2002k}. In K-anonymity, each user's identity is kept indistinguishable within a group of $k-1$ other users. On the other hand, Shokri et al.~\cite{shokri2011quantifying, shokri2011quantifying2} define the expected estimation error of the adversary as a metric to evaluate LPPMs. Ma et al.~\cite{ma2009location} use the uncertainty of the users' location
information to quantify the location privacy of the user in vehicular networks.
Li et al.~\cite{li2009tradeoff} define metrics to show the tradeoff between the privacy and utility of LPPMs.

Wang et al.\ tried to protect the privacy of the users for context sensing on smartphones, using Markov decision process (MDP)~\cite{wangprivacy}. The adversary's approach and user's privacy preserving mechanism is changing during time. The goal is to obtain the optimal policy of the users.

Previously, the mutual information has been used as a privacy metric in different topics,~\cite{salamatian2013hide, csiszar1996almost, calmon2015fundamental, sankar2013utility, yamamoto1983source}. However, in this paper we use the mutual information specifically for location privacy. For this reason, a new setting for this privacy problem will be provided and discussed. Specifically, we provide an information theoretic definition for location privacy using the mutual information. We show that wireless devices can achieve provable perfect location privacy by using the anonymization method in the suggested way.

In \cite{matching}, the author studies asymptotically optimal matching of sequences to source distributions. However, there are two key differences between \cite{matching} and this paper. First, \cite{matching} looks only at the optimal matching tests, but does not consider any privacy metric (i.e., perfect privacy) as considered in this paper. The major part of our work is to show that the mutual information converges to zero so we can conclude there is no privacy leak (hence perfect privacy). Also, the setting of \cite{matching} is different as it assumes a fixed distribution on sources (i.e., classical inference) as we assume the existence of a general (but possibly unknown) prior distributions for the sources (i.e. a Bayesian setting).

%%%%%%%%%%%%%%%%%%%%%%%%%%%%%%%%%%%%%%%%%%%%%%%%%%%%
\section{Framework}
\subsection{Defining Location Privacy}

 In the proposed framework, we consider a region in which a large number of wireless devices are using an LBS.  To support their location privacy, the anonymization technique is being used by the LBS. An outsider adversary $\mathcal{A}$ is interested in identifying users based on their locations and movements. We consider this adversary to be the strongest adversary that has complete statistical knowledge of the users' movements based on the previous observations or other resources. The adversary has a model that describes users' movements as a random process on the corresponding geographic area.

Let $X_u(t)$ be the location of user $u$ at time $t$, and $n$ be the number of users in our network. The location data of users can be represented in the form of the following stochastic processes:
\begin{center}
\begin{tabular}{ c c c c c c c}
$ X_1(1)$& $X_1(2)$& $X_1(3)$& $\cdots$& $X_1(m)$& $X_1(m+1)$& $\cdots$\\
$ X_2(1)$& $X_2(2)$& $X_2(3)$& $\cdots$& $X_2(m)$& $X_2(m+1)$& $\cdots$\\
$ X_3(1)$& $X_3(2)$& $X_3(3)$& $\cdots$& $X_3(m)$& $X_3(m+1)$& $\cdots$\\
$ \vdots$&$\vdots$&$ \vdots$& $\vdots$& $\vdots$ &$ \vdots$&$\vdots$\\
$ X_n(1)$& $X_n(2)$& $X_n(3)$& $\cdots$& $X_n(m)$& $X_n(m+1)$& $\cdots$\\
%$ X_{n+1}(1)$& $X_{n+1}(2)$& $X_{n+1}(3)$& $\cdots$& $X_{n+1}(m)$& $X_{n+1}(m+1)$& $\cdots$\\
$ \vdots$&$\vdots$&$ \vdots$& $\vdots$& $\vdots$ &$ \vdots$&$\vdots$
\end{tabular}
\end{center}
 The adversary's observations are anonymized versions of the $X_u(t)$'s produced by the anonymization technique. She is interested in knowing $X_u(t)$ for $u=1,2,...,n$ based on her $m$ anonymized observations for each of the $n$ users, where $m$ is a function of $n$, e.g., $m=m(n)$. Thus, at time $m$, the data shown in the box has been produced:

\begin{center}
\begin{tabular}{ c c c c c c c}
  \cline{1-5}
\multicolumn{1}{|c}{
$ X_1(1)$}& $X_1(2)$& $X_1(3)$& $\cdots$& \multicolumn{1}{c|}{$X_1(m)$}& $X_1(m+1)$& $\cdots$\\
\multicolumn{1}{|c}{
$ X_2(1)$}& $X_2(2)$& $X_2(3)$& $\cdots$& \multicolumn{1}{c|}{$X_2(m)$}& $X_2(m+1)$& $\cdots$\\
\multicolumn{1}{|c}{
$ X_3(1)$}& $X_3(2)$& $X_3(3)$& $\cdots$& \multicolumn{1}{c|}{$X_3(m)$}& $X_3(m+1)$& $\cdots$\\
\multicolumn{1}{|c}{
$ \vdots$}&$\vdots$&$ \vdots$& $\vdots$&\multicolumn{1}{c|}{ $\vdots$}&$ \vdots$&$\vdots$\\
\multicolumn{1}{|c}{
$ X_n(1)$}& $X_n(2)$& $X_n(3)$& $\cdots$& \multicolumn{1}{c|}{$X_n(m)$}& $X_n(m+1)$&
$\cdots$ \\  \cline{1-5}
%$ X_{n+1}(1)$& $X_{n+1}(2)$& $X_{n+1}(3)$& $\cdots$& $X_{n+1}(m)$& $X_{n+1}(m+1)$& $\cdots$\\
$ \vdots$&$\vdots$&$ \vdots$& $\vdots$& $\vdots$ &$ \vdots$&$\vdots$
\end{tabular}
\end{center}

 The goal of this paper is to find the function $m=m(n)$ in a way that perfect privacy is guaranteed. Let $\textbf{Y}^{(m)}$ be a collection of anonymized observations available to the adversary. That is, $\textbf{Y}^{(m)}$ is the anonymized version of the data in the box. We define \emph{perfect location privacy} as follows:

\begin{define}
User $u$ has perfect location privacy at time $t$, if and only if
\begin{align}%\label{}
\no  \lim_{n\rightarrow \infty} I \left(X_u(t);{\textbf{Y}^{(m)}}\right) =0,
\end{align}
where $I(X;Y)$ shows the mutual information between $x$ and $y$.
\end{define}
The above definition implies that over time, the adversary's anonymized observations do not give any information about the user's location. The assumption of large $n$, ($n \rightarrow \infty$), is valid for almost all applications that we consider since the numbers of users for such applications are usually very large.

In order to achieve perfect location privacy, we only consider anonymization techniques to confuse the adversary. In particular, the anonymization can be modeled as follows:

We perform a random permutation $\Pi^{(n)}$, chosen uniformly at random among all $n!$ possible permutations on the set of $n$ users, and then assign the pseudonym $\Pi^{(n)}(u)$ to user $u$
\begin{equation}
\nonumber  \Pi^{(n)}:\{1, 2,\cdots,n\} \mapsto \{1, 2, \cdots,n\}.
\end{equation}
Throughout the paper, we may use $\Pi(u)$ instead of $\Pi^{(n)}(u)$ for simplicity of the notation.

For $u=1,2,\cdots,n$, let \textbf{X}$_u^{(m)}$ be a vector which shows the $u^{th}$ user's locations at times $t={1, 2, \dots, m}$:
\[
  \textbf{X}_u^{(m)} = \left(X_{u}(1), X_{u}(2), \cdots, X_{u}(m)\right)^{T}
\]
Using the permutation function $\Pi^{(n)}$, the adversary observes a permutation of users' location vectors, $\textbf{X}_u^{(m)}$'s. In other words, the adversary observes

\begin{align}
\textbf{Y}^{(m)} &=\textrm{Perm}\left( \textbf{X}_{1}^{(m)}, \textbf{X}_{2}^{(m)}, \cdots,  \textbf{X}_{n}^{(m)}; \Pi^{(n)} \right) \\
\nonumber &=\left( \textbf{X}_{\Pi^{-1}(1)}^{(m)}, \textbf{X}_{\Pi^{-1}(2)}^{(m)}, \cdots,  \textbf{X}_{\Pi^{-1}(n)}^{(m)}\right)\\
\nonumber &=\left( \textbf{Y}_{1}^{(m)}, \textbf{Y}_{2}^{(m)}, \cdots,  \textbf{Y}_{n}^{(m)} \right) \ \ \\
\nonumber & \textbf{Y}_{u}^{(m)} = \textbf{X}_{\Pi^{-1}(u)}^{(m)} ,\  \textbf{Y}_{\Pi(u)}^{(m)} = \textbf{X}_{u}^{(m)}
\end{align}

\noindent where \textrm{Perm}(.) shows the applied permutation function. Then,
\begin{align}
%\nonumber
\no \textbf{Y}_{\Pi{(u)}}^{(m)}= \textbf{X}_u^{(m)} = \left( X_u(1), X_{u}(2), \cdots,  X_{u}(m)\right)^{T}.
\end{align}

\section{Example} \label{sec:example}
Here we provide a simple example to further elaborate the problem setting. Assume that we have only three users, $n = 3$, and five locations, $r = 5$, that users can occupy (Figure \ref{fig:example}). Also, let's assume that the adversary can collect $m(n)=4$ observations per user. Each user creates a path as below:
\[
\begin{tabular}{c| c}
    user  & path  \\   \hline
    user1 & $1\rightarrow 2\rightarrow3\rightarrow4$
    \\
    user2 &$2\rightarrow 1\rightarrow3\rightarrow5$
    \\
   user3 & $4\rightarrow 5\rightarrow1\rightarrow3$
\end{tabular}
\] \[
\begin{tabular}{c c c}
$\textbf{X}_1^{(4)} =  \begin{bmatrix}
 1 \\ 2 \\ 3 \\ 4
\end{bmatrix},$
&
$\textbf{X}_2^{(4)} = \begin{bmatrix}
 2 \\ 1 \\3 \\ 5
\end{bmatrix},$
&
$\textbf{X}_3^{(4)} = \begin{bmatrix}
  4 \\ 5 \\ 1 \\ 3
\end{bmatrix},$
\end{tabular}
\textbf{X}^{(4)} = \begin{bmatrix}
1 & 2 & 4 \\
2 & 1 & 5 \\
3 & 3 & 1 \\
4 & 5 & 3
\end{bmatrix}
\]

\begin{figure}
\centering
\includegraphics[width = 0.3\linewidth]{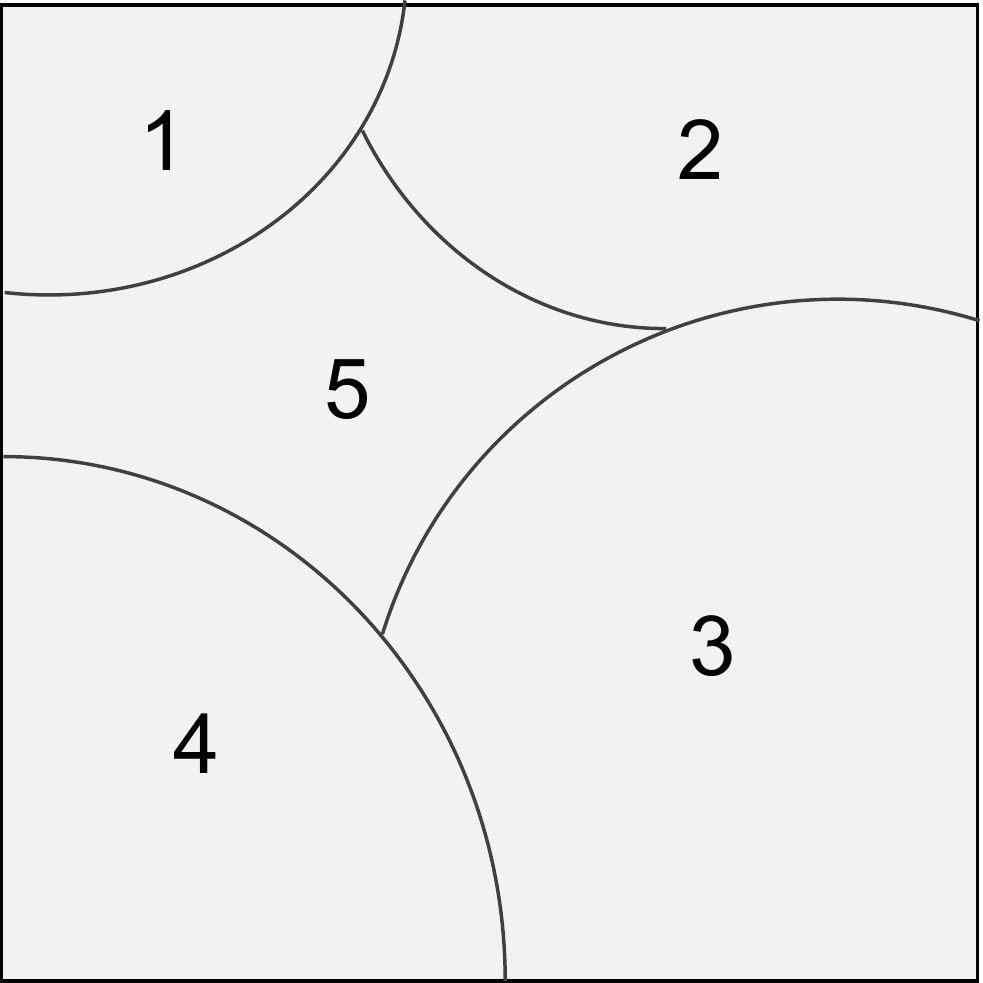}
\caption{An area is divided into five regions that users can occupy.}
\label{fig:example}
\end{figure}

To anonymize the users, we will assign a pseudonym to each. The pseudonyms are determined by the function defined by a random permutation on the user set:
\begin{align}
\no \Pi^{(3)} : \{1,2,3\} \mapsto \{1,2,3\}
\end{align}

For this example, suppose that the permutation function is given by $\Pi(1)=3$, $\Pi(2)=1$, and $\Pi(3)=2$.  The choice of the permutation is the only piece of information that is not available to the adversary.
So here, the adversary observes anonymized users and their paths:
\[
\begin{tabular}{c| c}
    pseudonym  & observation  \\   \hline
    user 1 & $ 2\rightarrow 1\rightarrow3\rightarrow5$
    \\
    user 2 & $ 4\rightarrow 5\rightarrow1\rightarrow3$
    \\
    user 3 & $ 1\rightarrow 2\rightarrow3\rightarrow4$
\end{tabular}
 \ \ \
 \begin{tabular}{c}
\\
$\textbf{Y}^{(4)} = \begin{bmatrix}
2 & 4 & 1\\
1 & 5 & 2\\
3 & 1 & 3 \\
5 & 3 & 4
\end{bmatrix}
$\end{tabular}
\]
and she wants to find which user (with the pseudomym $user 3$) actually made $1\rightarrow 2\rightarrow3\rightarrow4$, and so on for the other users. Based on the number of observations that the adversary collects for each user, $m = m(n) = 4$, and also the statistical knowledge of the users' movements, she aims at breaking the anonymization function and de-anonymizing the users. The accuracy of this method depends on the number of observations that the adversary collects, and thus our main goal in this paper is to find the function $m(n)$ in a way that the adversary is unsuccessful and the users have perfect location privacy.

%%%%%%%%%%%%%%%%%%%%%%%%%%%%%%%%%%%%%%%%%%%%%%%%%%%%%
\section{i.i.d Model}
\subsection{Perfect Location Privacy for a Simple Two-State Model}
To get a better insight about the location privacy problem, here we consider a simple scenario where there are only two states to which users can go, states $0$ and $1$. At any time $k \in \{0,1,2,\cdots \}$, user $u$ has probability $p_u \in (0,1)$ to be at state $1$, independent from her previous locations and other users' locations. Therefore,
\[ X_u(k) \sim Bernoulli(p_u).\]

To keep things general, we assume that $p_u$'s are drawn independently from some continuous density function, $f_P(p)$, on the $(0,1)$ interval. Specifically, there are $\delta_2>\delta_1>0$ such that\footnote{The condition $\delta_1 <f_P(p) <\delta_2$ is not actually necessary for the results and can be relaxed; however, we keep it here to avoid unnecessary technicalities.}
\begin{equation}
\no\begin{cases}
    \delta_1 <f_P(p) <\delta_2 & p \in (0,1)\\
    f_P(p)=0 &  p \notin (0,1)
\end{cases}
\end{equation}
The values of $p_u$'s are known to the adversary. Thus, the adversary can use this knowledge to potentially identify the users. Note that our results do not depend on the choice of $f_P(p)$ and we do not assume that we know the underlying distribution $f_P(p)$. All we assume here is the existence of such distribution. The following theorem gives a general condition to guarantee perfect location privacy:

\begin{thm}\label{two_state_thm}
For two locations with the above definition and anonymized observation vector of the adversary, $\textbf{Y}^{(m)}$, if all the following holds
\begin{enumerate}
 \item $m = cn^{2 -  \alpha}$,
 for some positive constants  $c$  and $\alpha<1$;
\item$p_1 \in (0,1)$;
\item $(p_2, p_3, \cdots, p_n) \sim f_P $;
\item $P = (p_1, p_2,\cdots, p_n)$ be known to the adversary;
\end{enumerate}
then, we have
\begin{align}
\no \forall  k \in \mathbb{N},\ \ \  \lim_{n \rightarrow \infty} I\left(X_1(k) ; \textbf{Y}^{(m)}\right) = 0
\end{align}
i.e., user $1$ has perfect location privacy.
\end{thm}

Note that although the theorem is stated for user 1, the symmetry of the problem allows it to be restated for all users. Also note that the theorem is proven for any $0<\alpha<1$. Therefore, roughly speaking, the theorem states that if the adversary obtains less than $O(n^2)$ observations per user, then all users have location privacy.

\subsection{The Intuition Behind The Proof}
Here we provide the intuition behind the proof, and the rigorous proof for Theorem~\ref{two_state_thm} is given in Appendix~\ref{sec:app}. Let us look from the adversary's perspective. The adversary is observing anonymized locations of the first user and she wants to figure out the index of the user that she is observing, in other words she wants to obtain $X_1(k)$ from $\textbf{Y}^{(m)}$. Note that the adversary knows the values of $p_1, p_2,\cdots, p_n$. To obtain $X_1(k)$, it suffices that the adversary obtains $\Pi(1)$. This is because $\textbf{X}_{1}^{(m)}= \textbf{Y}_{\Pi(1)}^{(m)}$, so
\begin{align}%\label{}
\no X_1(k)&= \textrm{the kth element of }\textbf{X}_{1}^{(m)}\\
\no &= \textrm{the kth element of }\textbf{Y}_{\Pi(1)}^{(m)}.
\end{align}

Since $X_u(k)$ is a Bernoulli random variable with parameter $p_u$, to do so, the adversary can look at the averages
\begin{align}%\label{}
\no \overline{Y}_{\Pi(u)}=\frac{Y_{\Pi(u)}(1)+Y_{\Pi(u)}(2)+...+Y_{\Pi(u)}(m)}{m}.
\end{align}
In fact, $\overline{Y}_{\Pi(u)}$'s provide sufficient statistics for this problem. Now, intuitively, the adversary is successful in recovering $\Pi(1)$ if two conditions hold:
\begin{enumerate}
  \item $\overline{Y}_{\Pi(1)} \approx p_1$.
  \item For all $u \neq 1$, $\overline{Y}_{\Pi(u)}$ is not too close to $p_1$.
\end{enumerate}

Now, note that by the Central Limit Theorem (CLT)
\begin{align}%\label{}
\no \frac{\overline{Y}_{\Pi(u)}-p_u}{\sqrt{\frac{p_u(1-p_u)}{m}}} \rightarrow N\left(0, 1\right).
\end{align}
That is, loosely speaking, we can write
\begin{align}%\label{}
\no \overline{Y}_{\Pi(u)} \rightarrow N\left(p_u, \frac{p_u(1-p_u)}{m}\right).
\end{align}

Consider an interval $I \in (0,1)$ such that $p_1$ falls into that interval and length of $I$, $l^{(n)}$, is chosen to be
\[ l^{(n)}=\frac{1}{n^{1-\eta}},\]
where $0<\eta<\frac{\alpha}{2}$ (Remember $\alpha$ was defined by the equation $m = cn^{2 -  \alpha}$ in the statement of Theorem \ref{two_state_thm}) . Note that $l^{(n)}$ goes to zero as $n$ becomes large. Also, note that for any $u \in {1,2,\cdots,n}$, the probability that $p_u$ is in $I$ is larger than $\delta_1 l^{(n)}$. In other words, since there are $n$ users, we can guarantee that a large number of $p_u$'s fall in $I$ since we have
\begin{align}%\label{}
\no  n l^{(n)}  \rightarrow \infty.
\end{align}
On the other hand, note that
\begin{align}%\label{}
\no \frac{\sqrt{\var(\overline{Y}_{\Pi(u)})}}{\textrm{length}(I)} &=\frac{\sqrt{\frac{p_u(1-p_u)}{m}}}{\frac{1}{n^{1-\eta}}}\\
\no &=n^{\frac{\alpha}{2}-\eta} \rightarrow \infty.
\end{align}

Note that here, we will have a large number of normal random variables $\overline{Y}_{\Pi(u)}$ whose expected values are in the interval $I$ (that has a vanishing length) with high probability and their standard deviations are much larger than the interval length (but equal to each other asymptotically, i.e. $\frac{p_1(1-p_1)}{m}$). Thus, distinguishing between them will become impossible for the adversary. In other words, the probability that the adversary will correctly identify $\Pi(1)$ goes to zero as $n$ goes to infinity. That is, the adversary will  most likely choose an incorrect value $j$ for $\Pi(1)$. In this case, since the locations of different users are independent, the adversary will not obtain any useful information by looking at $X_j(k)$. Of course, the above argument is only intuitive. The rigorous proof has to make sure all the limiting conditions work out appropriately. This has been accomplished in Appendix~\ref{sec:app}.

\subsection{Extension to $r$-States Model}
Here, we extend our results to a scenario in which we have $r \geq 2$ locations, $0,1, \cdots, r-1$. At any time $k \in \{0,1,2,\cdots \}$, user $u$ has probability $p_u(i) \in (0,1)$ to be at location $i$, independent from her previous locations and other users' locations. At any given time $k$, ${p}_u(i)$ shows the probability of user $u$ being at location $i$ and vector $\textbf{p}_u$ contains ${p}_u(i)$'s for all the locations
\[
{p}_u(i) = P(X_u(k)=i),
\]
\[
\textbf{p}_u = \left({p}_u(0), {p}_u(1), \cdots, {p}_u(r-1)\right).
\]
We assume that $p_u(i)$'s for $i=0,1,2,\cdots, r-2$ are drawn independently from some $r-1$ dimensional continuous density function $f_{\textbf{P}}(\textbf{p})$ on the $(0,1)^{r-1}$. In particular, define the range of distribution as
\begin{align}
\no  R_{\textbf{P}} = \{ (x_1, x_2, \cdots, x_{r-1}) \in (0,1)^{r-1}:  \ \ \ \ \ \ \ \ \ \ \ \ \ \ \ \ \\
\no  x_i > 0 , x_1+x_2+\cdots+x_{r-1} < 1\}.
\end{align}
Then, we assume there exists positive constants $\delta_1,\delta_2>0$ such that
\begin{equation}
\begin{cases}
\no    \delta_1 <f_{\textbf{P}}(\mathbf{p}_u) <\delta_2 & \textbf{p}_u \in R_{\textbf{P}}\\
    f_{\textbf{P}}(\mathbf{p}_u)=0 &  \textbf{p}_u \notin R_{\textbf{P}}
\end{cases}
\end{equation}
For example, Figure \ref{R_P} shows the range $R_{\textbf{P}}$ for the case where there are three locations, $r=3$.

\begin{figure}
  \centering
  \includegraphics[width=.4\linewidth, height=0.38 \linewidth]{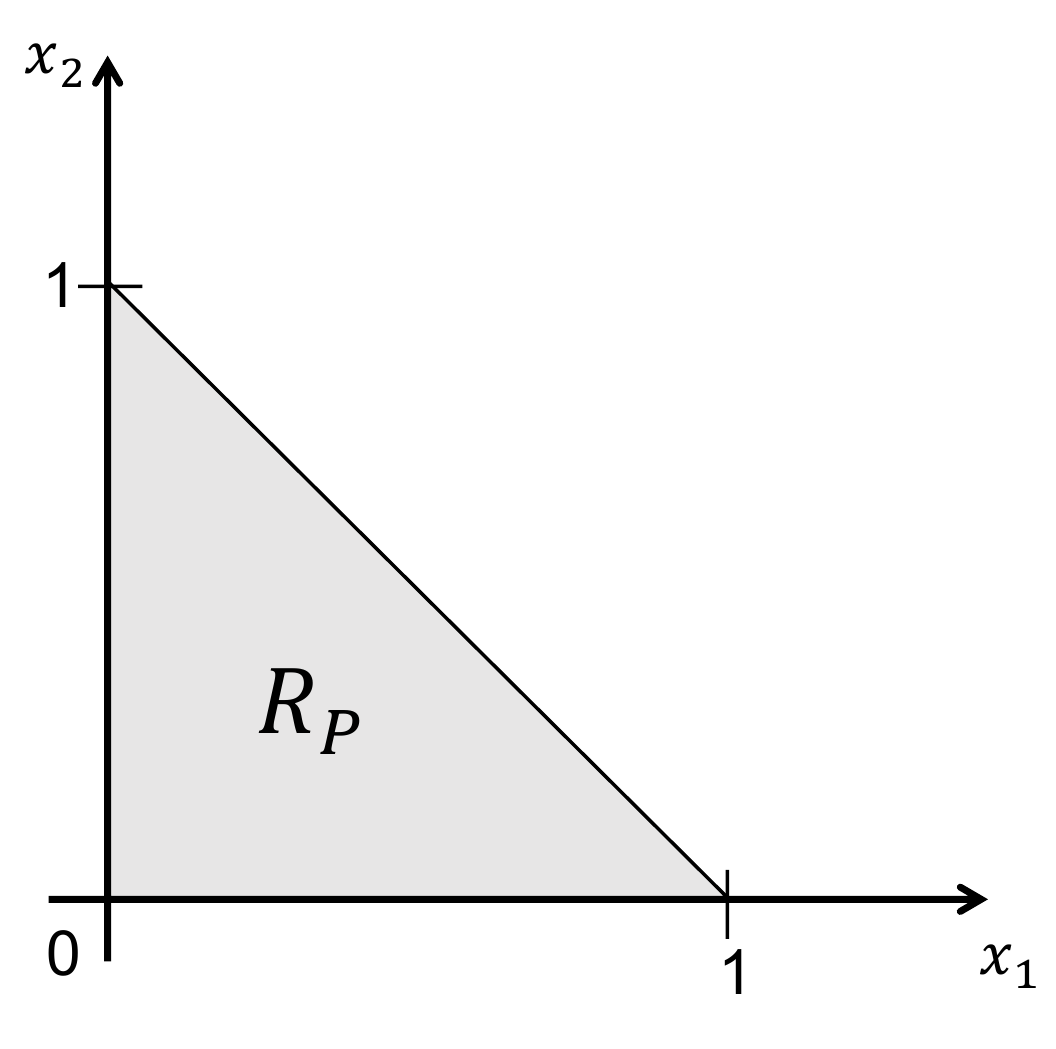}
  \centering
  \caption{$R_P$ for case $r=3, (d=2)$.}
  \label{R_P}
\end{figure}

\begin{thm}\label{r_state_thm}
For $r$ locations with the above definition and the adversary with an observation vector $\textbf{Y}^{(m)}$, if all the following holds
\begin{enumerate}
 \item $m = cn^{\frac{2}{r-1} -  \alpha}$,
 which $c,\alpha >0$ and are constant
\item$\textbf{p}_1 \in (0,1)^{r}$
\item $\left(\textbf{p}_2, \textbf{p}_3, \cdots, \textbf{p}_n\right) \sim f_{\textbf{P}} $, $0<\delta_1<f_{\textbf{P}}<\delta_2$
\item $\textbf{P} = \left(\textbf{p}_1, \textbf{p}_2,\cdots, \textbf{p}_n\right)$ is known to the adversary
\end{enumerate}

then, we have
\begin{align}
\no \forall  k \in \mathbb{N},\ \ \  \lim_{n \rightarrow \infty} I\left(X_1(k) ; \textbf{Y}^{(m)}\right) = 0
\end{align}
i.e., user $1$ has perfect location privacy.
\end{thm}

Proof of the Theorem~\ref{r_state_thm} is analogous to the proof of Theorem \ref{two_state_thm}. Here, we provide the general intuition. We do not provide the entire rigorous proof as it is for the most part repetition of the arguments provided for Theorem \ref{two_state_thm} in Appendix~\ref{sec:app}.

Let $d$ be equal to $r - 1$. As you can see in Figure~\ref{j_n}, for three locations there exists a set $J^{(n)}$ such that $\textbf{p}_1$ is in that set and we have
\[
  Vol(J^{(n)}) = (l^{(n)})^d.
\]
We choose $l^{(n)} = \frac{1}{n^{\frac{1}{d}-\eta}}$, where $\eta < \frac{\alpha}{2}$. Thus, the average number of users with $\textbf{p}$ vector in $J{(n)}$ is
\begin{align}
\no  n\frac{1}{\left(n^{\frac{1}{d}-\eta}\right)^d}= n^{d\eta} \rightarrow \infty \ \ \text{as} \ \ n \rightarrow \infty,
\end{align}
so we can guarantee that a large number of users are in the set $J^{(n)}$. This can be done exactly as in the proof of Theorem~\ref{two_state_thm} using Chebyshev's Inequality.

\begin{figure}
  \centering
  \includegraphics[width=.4\linewidth, height=0.38 \linewidth]{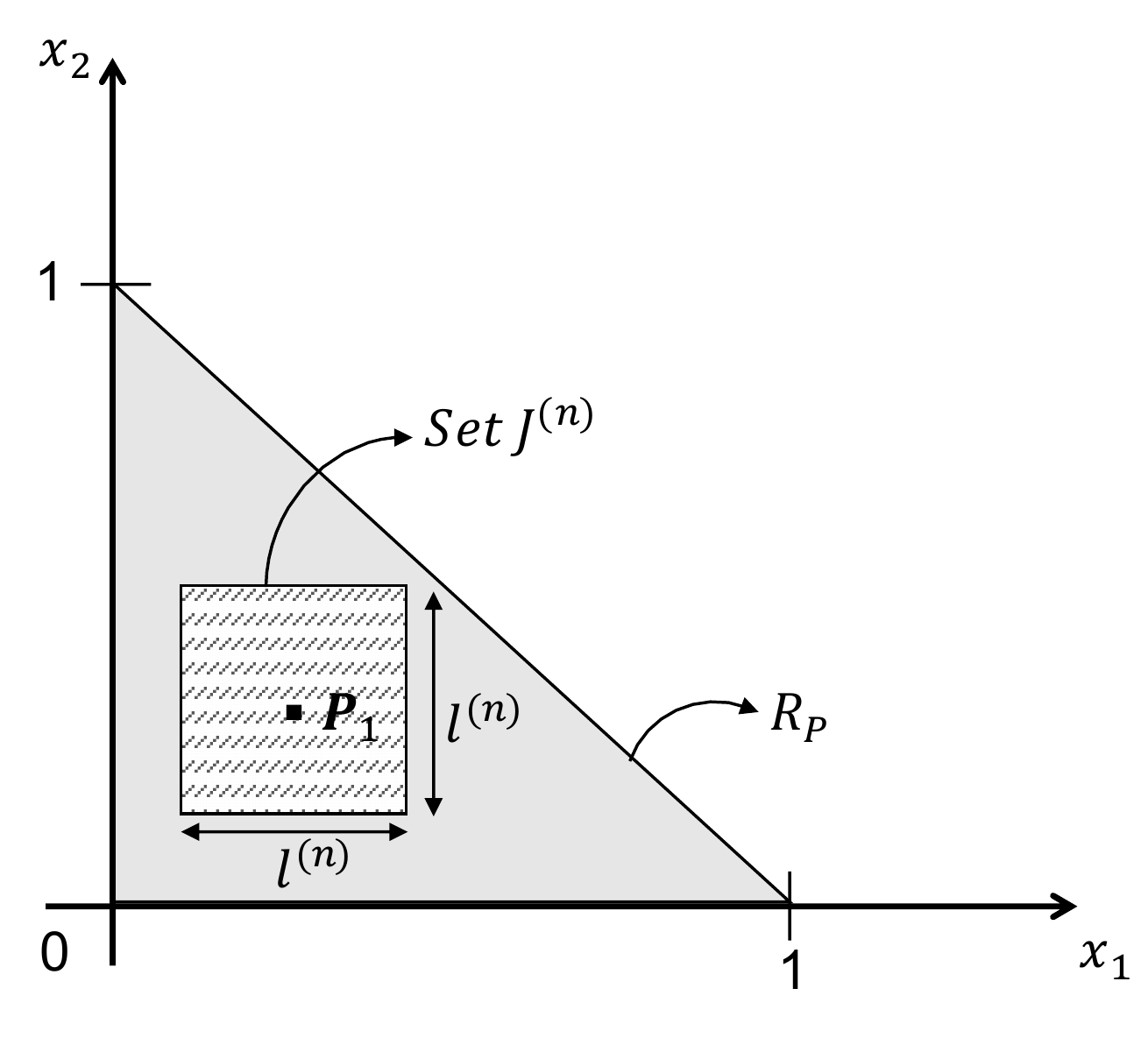}
  \centering
  \caption{$\textbf{p}_1 = ({p}_1(0), {p}_1(1), \cdots, {p}_1(r-1))$ is in set $J^{(n)}$ in $R_P$. }
  \label{j_n}
\end{figure}

Here, the number of times a user is at each location follows a multinomial distribution and in the long run, these numbers have a jointly gaussian distribution asymptotically as $n$ goes to infinity. The standard deviation of these variables are in the form of $ \frac{const}{\sqrt{m}}$. Moreover, the standard deviation over length of this interval is also large

\begin{align}
\no   \frac{\frac{const.}{\sqrt{m}}}{l^{(n)}}=\frac{const.\  n^{\frac{1}{d}-\eta}}{\sqrt{m}} \sim \frac{n^{\frac{1}{d}-\eta}}{(n^{\frac{2}{d}-\alpha})^{\frac{1}{2}}} = n^{\frac{\alpha}{2}-\eta} \rightarrow \infty  \ \ \text{as} \ \ n \rightarrow \infty.
\end{align}

Again, we have a large number of asymptotically jointly normal random variables that have a much larger standard deviation compared to the differences of their means. Thus, distinguishing between them becomes impossible.

This suggests that it is impossible for the adversary to correctly identify a user based on her observations even though he knows $\textbf{P}$ and $\textbf{Y}^{(m)}$. So, all the users have perfect location privacy. The proof can be made rigorous exactly the same way we did for the proof of Theorem~\ref{r_state_thm}, so we do not repeat the details here.

%%%%%%%%%%%%%%%%%%%%%%%%%%%%%%%%%%%%%
\section{Markov Chain Model}\label{sec:MC}
Assume there are $r$ possible locations to which users can go. We use a Markov chain with $r$ states to model movements of each user. We define $E$, the set of edges in this Markov chain, such that $(i,j)$ is in $E$ if there exists an edge from $i$ to $j$ with probability $p{(i,j)}>0$.

%\begin{figure}[H]
%\begin{center}
%\[
%\SelectTips {lu}{12 scaled 1500}
%\xymatrixcolsep{3pc}\xymatrixrowsep{3pc}\xymatrix{
%*++[o][F-]{1}  \ar@(dl,ul)[]^{p(1,1)} \ar@/^1pc/[r]^{p(1,2)}
%&*++[o][F-]{2}
%\\ *++[o][F-]{3} \ar@/^/[u]_{p(3,1)} &\cdots
%&*++[o][F-]{r} \ar@/^/[ul]_{p(r,2)}  \ar@/^/[ull]_{p(r,1)}
%}
%\]
%\caption{Markov chain model with $r$ states and $|E|$ edges.}
%\label{mc}
%\end{center}
%\end{figure}
We assume that this Markov structure chain gives the movement pattern of each user and what differentiates between users is their transition probabilities. That is, for fixed locations $i$ and $j$, two different users could have two different transition probabilities. For simplicity, let's assume that all users start at location (state) $1$, i.e., $X_u(1)=1$ for all $u=1,2,\cdots$. This condition is not necessary and can be easily relaxed; however, we assume it here for the clarity of exposition. We now state and prove the theorem that gives the condition for perfect location privacy for a user in the above setting.
\begin{thm}\label{main}
For an irreducible, aperiodic Markov chain with $r$ states and $|E|$ edges, if $m = cn^{\frac{2}{|E|-r}-\alpha}$, where $c>0$ and $\alpha>0$ are constants, then
\begin{align}
\lim_{n \rightarrow \infty}{I(X_1(k) ; \textbf{Y}^{(m)})} = 0,\ \ \ \forall k \in \mathbb{N},
\end{align}
i.e., user 1 has perfect location privacy.
\end{thm}

\begin{proof}

Let $M_u(i,j)$ be the number of observed transitions from state $i$ to state $j$ for user $u$. We first show that $M_{\Pi(u)}(i,j)$'s provide a sufficient statistic for the adversary when the adversary's goal is to obtain the permutation $\Pi^{(n)}$. To make this statement precise, let's define $\mathbf{M}_u^{(m)}$ as the matrix containing $M_u(i,j)$'s for user $u$:
\begin{equation}
\nonumber \mathbf{M}_u^{(m)}=
\begin{bmatrix}
  M_u(1,1) & M_u(1,2) & \cdots & M_u(1,r) \\
  M_u(2,1) & M_u(2,2) & \cdots & M_u(2,r) \\
  \cdots & \cdots &  \cdots & \cdots \\
  M_u(r,1) & M_u(r,2) & \cdots & M_u(r,r)\\
\end{bmatrix}
\end{equation}
Also, let $\mathbf{M}^{(m)}$ be the ordered collection of $\mathbf{M}_u^{(m)}$'s. Specifically,
\[
  \mathbf{M}^{(m)} = \left(\mathbf{M}_1^{(m)}, \mathbf{M}_2^{(m)}, \cdots, \mathbf{M}_n^{(m)}\right)
\]
The adversary can obtain $\textrm{Perm} \big(\mathbf{M}^{(m)}, \Pi^{(n)} \big)$, a permuted version of $\mathbf{M}^{(m)}$. In particular, we can write
\begin{align}
\no\textrm{Perm} \big(\mathbf{M}^{(m)}, \Pi^{(n)} \big)&=\textrm{Perm}\left( \mathbf{M}_1^{(m)}, \mathbf{M}_2^{(m)}, \cdots, \mathbf{M}_n^{(m)}; \Pi^{(n)} \right) \\
\nonumber &=\left( \mathbf{M}_{\Pi^{-1}(1)}^{(m)}, \mathbf{M}_{\Pi^{-1}(2)}^{(m)}, \cdots,  \mathbf{M}_{\Pi^{-1}(n)}^{(m)}\right).
\end{align}

We now state a lemma that confirms $\textrm{Perm} \big(\mathbf{M}^{(m)}, \Pi^{(n)} \big)$ is a sufficient statistic for the adversary, when the adversary's goal is to recover $\Pi^{(n)}$. Remember that $\textbf{Y}^{(m)}$ is the collection of anonymized observations of users' locations available to the adversary.

\begin{lem}\label{lem-Markov-suff}
  Given $\textrm{Perm} \big(\mathbf{M}^{(m)}, \Pi^{(n)} \big)$, the random matrix $\textbf{Y}^{(m)}$ and the random permutation $\Pi^{(n)}$ are conditionally independent. That is
 \begin{align}\label{eq:sufficient}
 P\left(\Pi^{(n)}=\pi \ \ \bigg{ | } \ \ \textbf{Y}^{(m)}=\mathbf{y}, \textrm{Perm} \big(\mathbf{M}^{(m)}, \Pi^{(n)} \big)=\mathbf{m} \right)  =
 \\
 \no P\left(\Pi^{(n)}=\pi \ \ \bigg{ | } \ \ \textrm{Perm} \big(\mathbf{M}^{(m)}, \Pi^{(n)} \big)=\mathbf{s} \right)
\end{align}

\end{lem}
Lemma \ref{lem-Markov-suff} is proved in the Appendix B.

Next note that since the Markov chain is irreducible and aperiodic, when we are determining $p(i,j)$'s, there are $d$ degrees of freedom, where $d$ is equal to $|E| - r$. This is because for each states $i$, we must have
\begin{align}
\no \sum \limits_{j = 1}^{r} p(i,j) = 1.
\end{align}
Thus, the Markov chain of the user $u$ is completely determined by $d$ values of $p(i,j)$'s which we show as
\begin{align}
\no \textbf{P}_{u} = \left(p_u(1),p_u(2),\cdots,p_u(d)\right)
\end{align}
and $\textbf{P}_{u}$'s are known to the adversary for all users. Note that the choice of $\textbf{P}_{u}$ is not unique; nevertheless, as long as we fix a specific $\textbf{P}_{u}$, we can proceed with the proof.
We define $E_d$ as the set of $d$ edges whose $p(i,j)$'s belong to  $\textbf{P}_{u}$. Let $R_{\textbf{p}} \subset \mathbb{R}^{d}$ be the range of acceptable values for $\textbf{P}_{u}$. For example, in Figure \ref{ex1} we have $|E|=6$ and $r=3$, so we have three independent transitions probabilities. If we choose $p_1$, $p_2$, and $p_3$ according to the figure, we obtain the following region
\begin{align}
\no R_{\textbf{p}} =\{(p_1,p_2,p_3)  \in \mathbb{R}^3:
 0 \leq p_i \leq 1 \textrm{ for } i = 1,2,3 \textrm{ and }\\ \no p_1 + p_2 \leq 1\}.
\end{align}

\begin{figure}[H]
\begin{center}
\[
\SelectTips {lu}{12 scaled 1500}
\xymatrixcolsep{5pc}\xymatrixrowsep{5pc}\xymatrix{
*++[o][F-]{1}  \ar@(dl,ul)[]^{p_1}\ar@/^/[r]^{p_2}  \ar@/^0.7pc/[dr]|{1-p_1-p_2}
& *++[o][F-]{2} \ar@/^1pc/[d]^{1}
\\ & *++[o][F-]{3}  \ar@/^/[u]_{p_3}  \ar@/^1.2pc/[ul]|{1-p_3}
}
\]
\caption{Three states Markov chain example}
\label{ex1}
\end{center}
\end{figure}
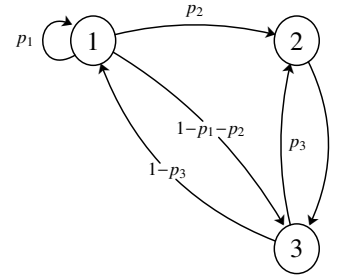
The statistical properties of each user are completely known to the adversary since she knows the Markov chain of each user. The adversary wants to be able to distinguish between users by having $m$ observations per user and also knowing $\textbf{P}_{u}$'s for all users.

In this model, we assume that $\textbf{P}_{u}$ for each user $u$ is drawn independently from a $d$-dimensional continuous density function, $f_{\textbf{P}}(\textbf{p})$. As before, we assume there exist positive constants $ \delta_1,\delta_2>0$, such that
\begin{equation}
\begin{cases}
\no    \delta_1 <f_{\textbf{P}}(\textbf{p}) <\delta_2 & \textbf{p} \in R_{\textbf{p}}\\
    f_{\textbf{P}}(\textbf{p})=0 &  \textbf{p} \notin R_{\textbf{p}}
\end{cases}
\end{equation}

We now claim that the adversary's position in this problem is mathematically equivalent to the the i.i.d model where the number of locations $r$ is equal to $d+1$ where $d = |E|-r$. First, note that since the Markov chain is irreducible and aperiodic, it has a unique stationary distribution which is equal to the limiting distribution. Next, define $\mathbf{Q}_u$ to be the vector consisting of all the transition probabilities of user $u$. In particular, based on the above argument, we can represent $\mathbf{Q}_u$ in the following way:
\begin{align}
\no \mathbf{Q}_u=[\mathbf{P}_{u}, \mathbf{P}_{u} \mathbf{B}],
\end{align}
where $\mathbf{B}$ is a non-random $d$ by $|E|-d$ matrix. Now, note that $\mathbf{P}_{u} \mathbf{B}$ is a non-random function of $\mathbf{P}_{u}$. In particular, if $M_u(i,j)$ shows the observed number transitions from state $i$ to state $j$ for user $u$, then we only need to know $M_u(i,j)$ for the edges in $E_d$, as the rest will be determined by the linear transform defined by $\mathbf{B}$. This implies that the decision problem for the adversary is reduced to the decision problem on transition probabilities in $\mathbf{P}_{u}$ and the adversary only needs to look at the $M_u(i,j)$'s for the edges in $E_d$. Now, this problem has exactly the same structure as the i.i.d model where the number of locations $r$ is equal to $d+1$ where $d = |E|-r$. In particular, $M_u(i,j)$'s have multinomial distributions and the statement of Theorem \ref{main} follows by applying Theorem \ref{r_state_thm}.

\end{proof}

\textit{Discussion:} One limitation of the above formulations is that all users must have the same Markov chain graphs with only different transition probabilities. In reality, there might be users that have different Markov chain graphs. For example, we might have users that never visit a specific region. Nevertheless, we can address this issue in the following way. If we are considering the location privacy of user $1$, we only consider users that have the same Markov chain graph (but with different transitions probabilities). The other users are easily distinguishable from user $1$ anyway. Now, $n$ would be the total number of this new set of users and again we can apply Theorem~\ref{main}. If $n$ is not large enough in this case, then we need to use location obfuscation techniques in addition to anonymization to achieve perfect privacy.

%%%%%%%%%%%%%%%%%%%%%%%%%%%%%%%
%\input{simulation}
%%%%%%%%%%%%%%%%%%%%%%%%%%

\section{Conclusion}

We presented an information theoretic definition for perfect location privacy using the mutual information between the users' actual locations and the anonymized observations that the adversary collects.
First, we modeled users' movements to be independent from their previous locations. In this model, we have $n$ users and $r$ locations. We prove that if the number of anonymized observations that the adversary collects, $m$, is less than $O(n^{\frac{2}{r-1}})$ then users will have perfect location privacy. So, if the anonymization method changes pseudonyms of users before $O(n^{\frac{2}{r-1}})$ observations is made by the adversary for each user, then the adversary cannot distinguish between users and they can achieve perfect location privacy.
Then, we modeled users' movements using Markov chains so that their current locations affect their next moves. We proved that for such a user, perfect location privacy is achievable if the pseudonym of the user is changed before $O(n^{\frac{2}{|E|-r}})$ number of observations is made by the adversary.

\appendices
\section{Proof of Theorem 1 (Perfect Location Privacy for Two-State Model)}\label{sec:app}
Here, we provide a formal proof for Theorem 1. In the proposed setting, we assume we have an infinite number of potential users indexed by integers, and at any step we consider a network consisting of  $n$ users, i.e., users $1, 2, \cdots, n$.  We would like to show perfect location privacy when $n$ goes to infinity. Remember that $X_u(t)$ shows the location of user $u$ at time $t$.

In a two-state model, let us assume we have state 0 and state 1. There is a sequence $p_1, p_2, p_3, \cdots$  for the users. In particular, for user $u$ we have
$p_u = P(X_u(k)=1)$ for times $k = 1,2,\cdots$.  Thus, the locations of each user $u$ are determined by a Bernoulli ($p_u$) process.
%These $p_u$'s are chosen according to $f_P(p)$ such that there exists $\delta_1,\delta_2>0$
%\begin{equation}
% \begin{cases}
% \no    \delta_1 <f_P(p) <\delta_2 & \forall p \in (0,1)\\
%     f_P(p)=0 &  \forall p \notin (0,1)
% \end{cases}
% \end{equation}
% Now the processes $\{X_u(k)\}_{k=1}^\infty$ are generated according to $p_u$'s for $u = 1,2,\cdots, n$.

When we set $n \in \mathbb{N}$ as the number of users, we assume $m$ to be the number of adversary's observations per user,
\[m = m(n) = cn^{2-\alpha} \ \ \ \textrm{where} \ \ 0<\alpha<1.\]
So, we have $n \rightarrow \infty$ if and only if $m \rightarrow \infty$.

As defined previously, $\textbf{X}_u^{(m)}$ contains $m$ number of user $u$'s locations and $\textbf{X}^{(m)}$ is the collection of $\textbf{X}_u^{(m)}$'s for all users,
\[\textbf{X}_u^{(m)} = \begin{bmatrix}
       X_u(1) \\ X_u(2) \\ \vdots \\X_u(m) \end{bmatrix} , \ \ \  \textbf{X}^{(m)} =\left( \textbf{X}_{1}^{(m)}, \textbf{X}_{2}^{(m)}, \cdots,  \textbf{X}_{n}^{(m)}\right).
\]

The permutation function applied to anonymize users is $\Pi^{(n)}$ (or simply $\Pi$). For any set $A \subset \{1,2,\cdots,n\} $, we define \[\Pi(A) = \{\Pi(u) : u \in A\}.\]

The adversary who knows all the $p_u$'s, observes $n$ anonymized users for $m$ number of times each and collects their locations in $\textbf{Y}^{(m)}$
\begin{align}
\no \textbf{Y}^{(m)} &=\textrm{Perm}\left(\textbf{X}_{1}^{(m)}, \textbf{X}_{2}^{(m)}, \cdots,  \textbf{X}_{n}^{(m)}; \Pi \right) \\
\nonumber &=\left( \textbf{Y}_{1}^{(m)}, \textbf{Y}_{2}^{(m)}, \cdots, \textbf{Y}_{n}^{(m)}\right ) \ \
\end{align}
where $\textbf{Y}_{u}^{(m)} = \textbf{X}_{\Pi^{-1}(u)}^{(m)} , \textbf{Y}_{\Pi(u)}^{(m)} = \textbf{X}_{u}^{(m)}$. % and remember that in the two-state model, $X_u(k)$'s are i.i.d and \[X_u(k) \sim Bernoulli(p_u).\]

Based on the assumptions of Theorem~\ref{two_state_thm}, if the following holds
\begin{enumerate}
 \item $m = cn^{2 -  \alpha}$,
 which $c>0, 0<\alpha<1$ and are constant
\item$p_1 \in (0,1)$
\item $(p_2, p_3, \cdots, p_n) \sim f_P $, $0<\delta_1<f_P<\delta_2$
\item $P = (p_1, p_2,\cdots, p_n)$ be known to the adversary,
\end{enumerate}
then we want to show
\begin{align}
\no \forall  k \in \mathbb{N},\ \ \  \lim_{n \rightarrow \infty} I\left(X_1(k) ; \textbf{Y}^{(m)}\right) = 0
\end{align}
i.e., user 1 has perfect location privacy and the same applies for all other users.

\subsection{Proof procedure}
Steps of the proof are as follows:
\begin{enumerate}
\item {We show that there exists a sequence of sets $J^{(n)} \subseteq \{1,2,\cdots,n\}$ with the following properties:
\begin{itemize}
  \item $1 \in J^{(n)}$
  \item if $N^{(n)} = |J^{(n)}|$ then, $N^{(n)} \rightarrow \infty$ as $n \rightarrow \infty$
  \item let $\{j_n\}_{n=1}^\infty$ be any sequence such that $j_n \in \Pi(J^{(n)})$ then
\begin{align}
\no    P\left(\Pi(1) = j_n |  \textbf{Y}^{(m)} , \Pi(J^{(n)})\right) \rightarrow 0
\end{align}

\end{itemize}
\item
{
We show that \[
  X_1(k)|\textbf{Y}^{(m)} , \Pi(J^{(n)}) \xrightarrow{d} Bernoulli(p_1).
  \]
}
\item
{
Using 2, we conclude \[
\no H\left(X_1(k)  |  \textbf{Y}^{(m)},  \Pi(J^{(n)})\right) \rightarrow H\left(X_1(k)\right)
\] and in conclusion,
\begin{align}
\no I\left(X_1(k)  ; \textbf{Y}^{(m)}\right) \rightarrow 0.
\end{align}
}
}
\end{enumerate}

\subsection{Detail of the proof}
We define $S_u^{(m)}$ for $u = 1,2,\cdots, n$ to be the number of times that user $u$ was at state 1,
\[
  S_u^{(m)} = X_u(1) + X_u(2) + \cdots +X_u(m).
\]
 Based on the assumptions, we have $S_u^{(m)} \sim Binomial(m,p_u)$. One benefit of $S_u^{(m)}$'s is that they provide a sufficient statistic for the adversary when the adversary's goal is to obtain the permutation $\Pi^{(n)}$. To make this statement precise, let's define $\mathbf{S}^{(m)}$ as the vector containing $S_u^{(m)}$, for $u=1,2,\cdots, n$:
\[
  \mathbf{S}^{(m)} = \left(S_1^{(m)}, S_2^{(m)}, \cdots, S_n^{(m)}\right)
\]
Note that
\begin{align}
\no S_{\Pi(u)}^{(m)} &= X_{\Pi(u)}(1) + X_{\Pi(u)}(2) + \cdots +X_{\Pi(u)}(m)\\
\no &= Y_{u}(1)+ Y_{u}(2)+ \cdots +Y_{u}(m)\ \ \ \ \textrm{ for }u=1, 2, \cdots, n.
\end{align}
Thus, the adversary can obtain $\textrm{Perm} \big(\mathbf{S}^{(m)}, \Pi^{(n)} \big)$, a permuted version of $\mathbf{S}^{(m)}$, by adding the elements in each column of $\textbf{Y}^{(m)}$. In particular, we can write
\begin{align}
\no\textrm{Perm} \big(\mathbf{S}^{(m)}, \Pi^{(n)} \big)&=\textrm{Perm}\left( S_1^{(m)}, S_2^{(m)}, \cdots, S_n^{(m)}; \Pi^{(n)} \right) \\
\nonumber &=\left( S_{\Pi^{-1}(1)}^{(m)}, S_{\Pi^{-1}(2)}^{(m)}, \cdots,  S_{\Pi^{-1}(n)}^{(m)}\right).
\end{align}

We now state and prove a lemma that confirms $\textrm{Perm} \big(\mathbf{S}^{(m)}, \Pi^{(n)} \big)$ is a sufficient statistic for the adversary when the adversary's goal is to recover $\Pi^{(n)}$. The usefulness of this lemma will be clear since we can use the law of total probability to break the adversary's decision problem into two steps of (1) obtaining the posterior probability distribution for $\Pi^{(n)}$ and (2) estimating the locations $X_u(k)$ given the choice of $\Pi^{(n)}$.

\begin{lem}\label{lem0}
  Given $\textrm{Perm} \big(\mathbf{S}^{(m)}, \Pi^{(n)} \big)$, the random matrix $\textbf{Y}^{(m)}$ and the random permutation $\Pi^{(n)}$ are conditionally independent. That is
 \begin{align}\label{eq:sufficient}
 P\left(\Pi^{(n)}=\pi \ \ \bigg{ | } \ \ \textbf{Y}^{(m)}=\mathbf{y}, \textrm{Perm} \big(\mathbf{S}^{(m)}, \Pi^{(n)} \big)=\mathbf{s} \right)  =
 \\ \no P\left(\Pi^{(n)}=\pi \ \ \bigg{ | } \ \ \textrm{Perm} \big(\mathbf{S}^{(m)}, \Pi^{(n)} \big)=\mathbf{s} \right)
\end{align}

\end{lem}

\begin{proof}

Remember
\begin{align}
\nonumber \textbf{Y}^{(m)} &=\textrm{Perm}\left( \textbf{X}_{1}^{(m)}, \textbf{X}_{2}^{(m)}, \cdots,  \textbf{X}_{n}^{(m)}; \Pi^{(n)} \right) \\
\nonumber &=\left( \textbf{X}_{\Pi^{-1}(1)}^{(m)}, \textbf{X}_{\Pi^{-1}(2)}^{(m)}, \cdots,  \textbf{X}_{\Pi^{-1}(n)}^{(m)}\right).
\end{align}
Note that $\textbf{Y}^{(m)}$ (and therefore $\mathbf{y}$) is an $m$ by $n$ matrix, so we can write
\begin{align}
\nonumber \mathbf{y} &=\left( \textbf{y}_{1}, \textbf{y}_{2}, \cdots,  \textbf{y}_{n} \right),
\end{align}
where for $u=1,2,\cdots, n$, we have
\[\mathbf{y}_u = \begin{bmatrix}
\nonumber       y_u(1) \\ y_u(2) \\ \vdots \\y_u(m) \end{bmatrix}.
\]
Also, $\mathbf{s}$ is a $1$ by $n$ vector, so we can write

\begin{align}
\nonumber \mathbf{s}&=\left( s_{1}, s_{2}, \cdots, s_{n} \right).
\end{align}

We now show that the two sides of Equation \ref{eq:sufficient} are equal. The right hand side probability can be written as

 \begin{align}%\label{eq:suff}
 \no P\left(\Pi^{(n)}=\pi \ \ \bigg{ | } \ \ \textrm{Perm} \big(\mathbf{S}^{(m)}, \Pi^{(n)}\big)=\mathbf{s} \right)=\ \ \ \ \ \ \ \ \ \ \ \ \ \ \ \ \\
  \no \frac{P\left( \textrm{Perm} \big(\mathbf{S}^{(m)}, \Pi^{(n)}\big)=\mathbf{s} \ \ \bigg{ | } \ \ \no \Pi^{(n)}=\pi  \right) P\left( \Pi^{(n)}=\pi  \right)}{P \bigg(\textrm{Perm} \big(\mathbf{S}^{(m)}, \Pi^{(n)}\big)=\mathbf{s}\bigg)}=&\\
\no \frac{P\left( \textrm{Perm} \big(\mathbf{S}^{(m)}, \pi \big)=\mathbf{s} \ \ \bigg{ | } \ \ \no \Pi^{(n)}=\pi  \right)} { n! P \bigg(\textrm{Perm} \big(\mathbf{S}^{(m)},\Pi^{(n)}\big)=\mathbf{s} \bigg)}=& \ \ \ \ \ \\
\no \frac{P\left( \textrm{Perm} \big(\mathbf{S}^{(m)}, \pi \big)=\mathbf{s} \right)} { n! P \bigg(\textrm{Perm} \big(\mathbf{S}^{(m)},\Pi^{(n)}\big)=\mathbf{s} \bigg)}.
\end{align}
Now note that
 \begin{align}%\label{eq:suff}
 \no P\left( \textrm{Perm} \big(\mathbf{S}^{(m)}, \pi \big)=\mathbf{s} \right) &=P \left( \bigcap_{j=1}^{n} \left( S_{\pi^{-1}(j)}^{(m)}=s_j\right) \right)\\
 \no &=P \left( \bigcap_{u=1}^{n} \left( S_{u}^{(m)}=s_{\pi (u)}\right) \right)\\
 \no &=\prod_{u=1}^{n} P\left( S_{u}^{(m)}=s_{\pi (u)}\right)\\
 \no &= \prod_{u=1}^{n}   {m \choose s_{\pi (u)}} p_u^{s_{\pi (u)}}(1-p_u)^{m-s_{\pi (u)}}\\
 \no &= \prod_{k=1}^{n}   {m \choose s_k} \prod_{u=1}^{n}   p_u^{s_{\pi (u)}}(1-p_u)^{m-s_{\pi (u)}}
\end{align}
Similarly, we obtain
\begin{align}%\label{eq:suff}
 \no P \bigg(\textrm{Perm} \big(\mathbf{S}^{(m)},\Pi^{(n)}\big)=\mathbf{s} \bigg) = \ \ \ \ \ \ \ \ \ \ \ \ \ \ \ \ \
 \\ \no \sum_{\textrm{all permutations }\pi'} P\left( \textrm{Perm} \big(\mathbf{S}^{(m)}, \pi' \big)=\mathbf{s} \bigg{ | }  \no \Pi^{(n)}=\pi'  \right) P\left( \Pi^{(n)}=\pi'  \right)\\
 \no = \frac{1}{n!} \sum_{\textrm{all permutations }\pi'}  \prod_{k=1}^{n}   {m \choose s_k}  \prod_{u=1}^{n}  p_u^{s_{\pi' (u)}}(1-p_u)^{m-s_{\pi' (u)}}\\
 \no = \frac{1}{n!} \prod_{k=1}^{n}   {m \choose s_k}  \sum_{\textrm{all permutations }\pi'}   \prod_{u=1}^{n}  p_u^{s_{\pi' (u)}}(1-p_u)^{m-s_{\pi' (u)}}.
\end{align}
Thus, we conclude that the right hand side of Equation \ref{eq:sufficient} is equal to
 \begin{align}%\label{eq:suff}
 \no  \frac{\prod_{u=1}^{n}   p_u^{s_{\pi (u)}}(1-p_u)^{m-s_{\pi (u)}}} {\sum_{\textrm{all permutations }\pi'}   \prod_{u=1}^{n}  p_u^{s_{\pi' (u)}}(1-p_u)^{m-s_{\pi' (u)}}}.
\end{align}

Now let's look at the left hand side of Equation \ref{eq:sufficient}. First, note that in the left hand side probability in Equation \ref{eq:sufficient} we must have
\begin{align}\label{eq:y-s}
 s_u&=\sum_{k=1}^{m} y_u(k) \ \ \ \ \textrm{ for }u=1, 2, \cdots, n.
\end{align}
Next, we can write
\begin{align}
 \no P\left(\Pi^{(n)}=\pi \ \ \bigg{ | } \ \ \textbf{Y}^{(m)}=\mathbf{y}, \textrm{Perm} \big(\mathbf{S}^{(m)}, \Pi^{(n)} \big)=\mathbf{s} \right)  =
 \\ \no P\left(\Pi^{(n)}=\pi \ \ \bigg{ | } \ \ \textbf{Y}^{(m)}=\mathbf{y} \right).
\end{align}
This is because $\textrm{Perm} \big(\mathbf{S}^{(m)}, \Pi^{(n)} \big)$ is a function of $\textbf{Y}^{(m)}$. We have

 \begin{align}
 \no P\left(\Pi^{(n)}=\pi \ \ \bigg{ | } \ \ \textbf{Y}^{(m)}=\mathbf{y} \right)=\ \ \ \ \ \ \
 \\ \no \frac{ P\left(\textbf{Y}^{(m)}=\mathbf{y}\ \ \bigg{ | } \ \  \Pi^{(n)}=\pi   \right) P\left(  \Pi^{(n)}=\pi \right)}{P\left(\textbf{Y}^{(m)}=\mathbf{y} \right)}\ \ \ \ \ \
\end{align}

We have
 \begin{align}
 \no P\left(\textbf{Y}^{(m)}=\mathbf{y}\ \ \bigg{ | } \ \  \Pi^{(n)}=\pi   \right) \ \ \ \ \ \ \ \ \ \ \ \ \ \ \ \ \ \ \ \ \ \ \ \ \ \ \ \
 \\ =\no \prod_{u=1}^{n} p_u^{\sum_{k=1}^{m} y_{\pi(u)}(k)} (1-p_u)^{m-\sum_{k=1}^{m} y_{\pi(u)}(k)}  \ \ \ \ \ \ \ \ \\
 \no  =\prod_{u=1}^{n} p_u^{s_{\pi(u)}} (1-p_u)^{m-s_{\pi(u)}} \ \ \textrm{Using Euqation (\ref{eq:y-s})}
\end{align}
Similarly, we obtain
 \begin{align}
 \no P\left(\textbf{Y}^{(m)}=\mathbf{y} \right)&= \frac{1}{n!}
  \sum_{\textrm{all permutations }\pi'} \prod_{u=1}^{n} p_u^{s_{\pi'(u)}} (1-p_u)^{m-s_{\pi'(u)}}
\end{align}
Thus, we conclude that the left hand side of Equation \ref{eq:sufficient} is equal to
 \begin{align}%\label{eq:suff}
 \no  \frac{\prod_{u=1}^{n}   p_u^{s_{\pi (u)}}(1-p_u)^{m-s_{\pi (u)}}} {\sum_{\textrm{all permutations }\pi'}   \prod_{u=1}^{n}  p_u^{s_{\pi' (u)}}(1-p_u)^{m-s_{\pi' (u)}}},
\end{align}
which completes the proof.

\end{proof}

Next, we need to turn our attention to defining the critical set $J^{(n)}$. First, remember that
 \[m = cn^{2-\alpha}  \ \ \ \ \textrm{where} \ \ \ \ 0<\alpha<1.\]
We choose real numbers $\theta$ and $\phi$ such that
$
  0<\theta<\phi<\frac{\alpha}{2-\alpha}
$, and define
\begin{align}
\no  \epsilon_m \triangleq \frac{1}{m^{\frac{1}{2}+\phi}} \ \ \ \beta_m \triangleq \frac{1}{m^{\frac{1}{2}-\theta}}.
\end{align}
We now define the set $J^{(n)}$ for any positive integer $n$ as follows: Set $J^{(n)}$ consists of the indices of users such that the probability of them being at state 1 is within a range with $\epsilon_m$ difference around $p_1$,
\[
  J^{(n)} = \{  i \in \{1, 2, \dots, n\} : p_1-\epsilon_m <p_i<p_1+\epsilon_m\}.
\]
Clearly for all $n, 1 \in J^{(n)} $. The following lemma confirms that the number of elements in $J^{(n)}$ goes to infinity as $n \rightarrow \infty$.

\begin{lem}
  \label{lem1}
  If $N^{(n)} \triangleq |J^{(n)}|$, then $N^{(n)} \rightarrow \infty$ as $n \rightarrow \infty$.
  More specifically, as $n \rightarrow \infty$,
  \[
    \exists \lambda , c''>0: \ \ \ P(N^{(n)} > c''n^\lambda) \rightarrow 1 \hspace{20pt} \textrm{ as } \hspace{5pt}n \rightarrow \infty.
  \]
\end{lem}
\begin{proof} Remember that we assume $p_u$'s are drawn independently from some continuous density function, $f_P(p)$, on the $(0,1)$ interval which satisfies
\begin{equation}
\no\begin{cases}
    \delta_1 <f_P(p) <\delta_2 & p \in (0,1)\\
    f_P(p)=0 &  p \notin (0,1)
\end{cases}
\end{equation}
  So given $p_1 \in (0,1)$, for $n$ large enough (so that $\epsilon_m$ is small enough), we have
  \[
    P( p_1-\epsilon_m <p_i<p_1+\epsilon_m )  = \int_{p_1-\epsilon_m}^{ p_1+\epsilon_m}f_P(p) dp,
    \]
    so we can conclude that
  \[
  2\epsilon_m\delta_1<P( p_1-\epsilon_m <p_i<p_1+\epsilon_m ) <2\epsilon_m\delta_2.
  \]
  We can find a $\delta$ such that $\delta_1<\delta<\delta_2$ and
\[
P( p_1-\epsilon_m <p_i<p_1+\epsilon_m ) = 2\epsilon_m\delta.
\]
Then, we can say that $N^{(n)} \sim Binomial(n,2\epsilon_m\delta)$, where
\[ \epsilon_m = \frac{1}{m^{\frac{1}{2}+\phi}}  =  \frac{1}{({cn^{2-\alpha}})^{(\frac{1}{2}+\phi)}}. \]
The expected value of $N^{(n)}$ is $n2\epsilon_m\delta$, and by substituting $\epsilon_m$ we get
\[
  E[N^{(n)}] = n2\epsilon_m\delta = \frac{n2\delta}{({c'n^{2-\alpha}})^{(\frac{1}{2}+\phi)}} = c''n^{(\frac{\alpha}{2}+\alpha\phi-2\phi)}.
\]
Let us set $\lambda = \frac{\alpha}{2}+\alpha\phi-2\phi$.
Since $\phi < \frac{\alpha}{2-\alpha}$ , we have $\lambda > 0$. Therefore, we can write
\[
  E[N^{(n)}] = c''n^\lambda,
\]
\[
  Var(N^{(n)}) = n(2\epsilon_m\delta)(1-2\epsilon_m\delta) \rightarrow n^\lambda (1+o(1)).
\]
Using Chebyshev's inequality
\[
  P(|N^{(n)} -   E[N^{(n)}] | > \frac{c''}{2}n^\lambda) < \frac{n^\lambda(1+o(1))}{\frac{c''^2}{4}n^{2\lambda}} \rightarrow 0
\]
\[
  P(N^{(n)} > \frac{c''}{2}n^\lambda) \rightarrow 1 \ \ \ \ \ \textrm{as}\ \ n \rightarrow \infty.
\]
\end{proof}

The next step in the proof is to show that users that are identified by the set $J^{(n)}$ produce a very similar moving process as user 1. To make this statement precise, we provide the following definition. Define the set $A^{(m)}$ as the interval in $\mathbb{R}$ consisting of real numbers which are within the $m\beta_m$ distance from $mp_1$ (the expected number of times that user 1 is at state 1 during the $m$ number of observations),
\[
 A^{(m)} = \{x \in R , m(p_1 -\beta_m)  \leq x  \leq m(p_1 +\beta_m)\}.
\]
\begin{lem}\label{lem2} We have
\[  P\left(\bigcap_{j \in J^{(n)}} \left( S_j^{(m)} \in A^{(m)} \right) \right) \rightarrow 1 \ \ \ \ \textrm{as}\ \ n \rightarrow \infty
\]
\end{lem}
\begin{proof}
  Let $ j \in J^{(n)}$ and $p_1-\epsilon_m <p_j<p_1+\epsilon_m$. Then
  $ S_j^{(m)} \sim Binomial(m,p_j).$  By Large Deviation Theory (Sanvo's Theorem), we can write
  \[
  P\left(S_j^{(m)} >m(p_1+\beta_m)\right) < \]
  \[(m+1) 2^{-mD\left(Bernoulli(p_1+\beta_m)\|Bernoulli(p_j)\right)}
  \]
  By using the fact that for all $p \in (0,1)$
  \[
    D\left(Bernoulli(p+\epsilon)\|Bernoulli(p)\right) = \frac{\epsilon^2}{2p(1-p)\ln p} + O(\epsilon^3),
  \]
 we can write
\begin{align}
 \no  D\left(Bernoulli(p_1+\beta_m)\|Bernoulli(p_j)\right) =\ \ \ \ \ \ \
\\ \no   \frac{(p_1+\beta_m-p_j)^2}{2p_j(1-p_j)\ln2}+O\left((p_1+\beta_m-p_j)^3\right).
\end{align}
  Note that $
  \left|p_1 - p_j\right| < \epsilon_m$, so for large $m$ we can write \[ \left|p_1+\beta_m-p_j\right| \ge \beta_m-\epsilon_m = \frac{1}{m^{\frac{1}{2}+\phi}} -\frac{1}{m^{\frac{1}{2}-\theta}} > \frac{\frac{1}{2}}{m^{\frac{1}{2}-\theta}}.
\]
 so we can write
\begin{align}
\no D\left(Bernoulli(p_1+\beta_m)\|Bernoulli(p_j)\right) =  \ \ \ \ \ \
 \\  \no \frac{1}{8p_j(1-p_j)m^{1-2\theta} \ln 2 } + O((p_1+\beta_m-p_j)^3)
\end{align}
 and for some constant $c'>0$
 \[D\left(Bernoulli(p_1+\beta_m)\| Bernoulli(p_j)\right)
  > \frac{c'}{m^{1-2\theta}} \Rightarrow \]\[ m D\left(Bernoulli(p_1+\beta_m)\|Bernoulli(p_j)\right)
   > \frac{mc'}{m^{1-2\theta}} > c'm^{2\theta} \Rightarrow
 \]
 \[
    P(S_j^{(m)} >m(p_1+\beta_m)) < m2^{-c'm^{2\theta}}.
 \]
 So in conclusion
 \[
    P\left(\bigcup_{j\in J^{(n)}}S_j^{(m)} >m(p_1+\beta_m)\right) < |J^{(n)}| m2^{-c'm^{2\theta}}
    \]\[|J^{(n)}| m2^{-c'm^{2\theta}}< nm2^{-c'm^{2\theta}} < m^22^{-c'm^{2\theta}}\rightarrow 0 \ \ \ \textrm{as} \ \ m \rightarrow \infty.
 \]

Similarly we obtain
\[
    P\left(\bigcup_{j\in J^{(n)}}S_j^{(m)} <m(p_1-\beta_m)\right) \rightarrow 0 \ \ \ \textrm{as} \ \ m \rightarrow \infty,
 \]
which completes the proof. This shows that for all users $j$  for which $p_j$ is within $2\epsilon$ range around $p_1$, i.e. it is in set $J^{(n)}$, the average number of times that this user was at state 1 is within $m\beta_m$ from $mp_1$ with high probability.
\end{proof}

We are now in a position to show that distinguishing between the users in $J^{(n)}$ is not possible for an outside observer (i.e., the adversary) and this will pave the way in showing perfect location privacy.
\begin{lem}\label{lem3}
  Let $\{a_m\}_{m=1}^\infty$, $\{b_m\}_{m=1}^\infty$ be such that $a_m, b_m$ are in set $A^{(m)}$ and also $\{i_m\}_{m=1}^\infty$, $\{j_m\}_{m=1}^\infty$ be such that $i_m , j_m$ are in set $J^{(n)}$. Then, we have
  \[
    \frac{P(S_{i_m}^{(m)} = a_m , S_{j_m}^{(m)} = b_m)}{P(S_{i_m}^{(m)} = b_m , S_{j_m}^{(m)} = a_m)} \rightarrow 1 \ \ \ \textrm{as}  \ \ m\rightarrow \infty.
  \]
\end{lem}
\begin{proof}
  Remember that
  \[
    A^{(m)} = \{x \in R, m(p_1 -\beta_m)  \leq x  \leq m(p_1 +\beta_m)\}
  \]
  where $\beta_m = \frac{1}{m^{\frac{1}{2}-\theta}}$ and $S_j^{(m)} \sim Binomial(m,p_j)$. Thus,
  $
    S_{i_m}^{(m)} \sim Binomial(m,p_{i_m})$ and $S_{j_m}^{(m)} \sim Binomial(m,p_{j_m}),
  $
  \[
    P(S_{i_m}^{(m)} = a_m) = \dbinom{n}{a_m} p_{i_m}^{a_m} (1- p_{i_m})^{m-a_m},
  \]

  \[
    P(S_{j_m}^{(m)} = b_m) = \dbinom{n}{b_m} p_{j_m}^{b_m} (1- p_{j_m})^{m-b_m}.
  \] In conclusion,
  \[
    \Delta_m = \frac{P(S_{i_m}^{(m)} = a_m , S_{j_m}^{(m)} = b_m)}{P(S_{i_m}^{(m)} = b_m , S_{j_m}^{(m)} = a_m)} = (\frac{p_{i_m}}{p_{j_m}})^{a_m - b_m} (\frac{1-p_{j_m}}{1-p_{i_m}})^{a_m - b_m}
  \]
  \[
  \ln \Delta_m = (a_m - b_m ) \ln (\frac{p_{i_m}}{p_{j_m}}) + (a_m - b_m ) \ln (\frac{1-p_{j_m}}{1-p_{i_m}})
  \]
  and since $ \{i_m , j_m\} \in J^{(n)}$ we have
  \[
    \left|p_{i_m} -p_{j_m}\right| \leq 2 \epsilon_m =  \frac{2}{m^{\frac{1}{2}+\phi}}.
  \]
  Also, since $\{a_m,b_m\} \in A^{(m)}$ we can say that
  \[
    \left|a_m - b_m\right| \leq 2m\beta_m.
  \]
  Since $p_{i_m} \leq p_{j_m} + 2\epsilon_m$ and $1-p_{j_m} \leq (1-p_{i_m}) + 2\epsilon_m$ and \[\ln (1+\epsilon_m) = \epsilon_m + O(\epsilon_m^2)\] we can write
  \[
  \ln \Delta_m \leq 2m\beta_m \epsilon_m +2m\beta_m \epsilon_m + 2m\beta_m O(\epsilon_m^2)
  \]
  and since $\phi > \theta$,
  \[
    m\beta_m \epsilon_m =  m \frac{1}{m^{\frac{1}{2}+\phi}} \frac{1}{m^{\frac{1}{2}-\theta}} = \frac{1}{m^{\phi - \theta}} \rightarrow 0,
  \
  \]
  \[
    \Rightarrow \ln \Delta_m \rightarrow 0
  \]
  \[ \Rightarrow \Delta_m \rightarrow 1.\]
  Note that the convergence is uniform.

This shows that for two users $i$ and $j$, if the probability of them being at state 1 is in set $J^{(n)}$, $p_i,p_j \in J^{(n)}$, and also the observed number of times for these users to be at state 1 is in set $A^{(m)}$, then distinguishing between these two users is impossible.
\end{proof}

\begin{lem}\label{lem4}
  For any $j \in \Pi(J^{(n)})$, we define $W_{j}^{(n)}$ as follows
  \[
  W_{j}^{(n)} = P(\Pi(1) = j |  \textbf{Y}^{(m)} , \Pi(J^{(n)})).
  \]
  Then, for all $j^{(n)} \in \Pi(J^{(n)})$,
  \[
    N^{(n)} W_j^{(n)} \xrightarrow{p} 1.
  \]
More specifically, for all $\gamma_1,\gamma_2>0$ , there exists $n_o$ such that if $n > n_o$:
\[\forall j \in \Pi(J^{(n)}):\ \
  P\left( \left| N^{(n)} W_j^{(n)} - 1 \right| > \gamma_1 \right) < \gamma_2.
\]
\end{lem}
\begin{proof}
  This is the result of Lemma~\ref{lem3}. First, remember that
  \[
    \sum_{j \in  \Pi(J^{(n)})} W_{j}^{(n)} = 1,
  \]
  and also note that
  \[
    |\Pi(J^{(n)})| = |J^{(n)}| = N^{(n)} \rightarrow \infty \ \ \ \textrm{as} \ \ n \rightarrow \infty.
  \]
  Here, we show that for any  $\{j_{n}\}_{n=1}^\infty \in \Pi(J^{(n)})$,
  \[
    \frac{W_{j_n}^{(n)}}{W_{1}^{(n)}} = \frac{P\left(\Pi(1) = j |D\right)}{P\left(\Pi(1) = 1 |D\right)} \xrightarrow{p} 1
  \]
  where $ D = \left(\textbf{Y}^{(m)} , \Pi(J^{(n)})\right)$.

  Let $a_i$, for $ i \in \Pi(J^{(n)})$, be the permuted observed values of $S^{(m)}_i$'s. Then note that
  \[
    P(\Pi(1) = j |D) = \sum_{\substack{\textrm{permutation}\\ \textrm{such that } \Pi(1) = j }} \sum_{ i \in \Pi(J) } P(S_i^ {(m)} = a_i ).
  \]
  Then, in
  \[
    \frac{W_{j_{n}}^{(n)}}{W_{1}^{(n)}} = \frac{P(\Pi(1) = j |D)}{P(\Pi(1) = 1 |D)}
  \]
  the numerator and denominator have the same terms. In particular, for each term
  \[
    P(S_j^{(m)} = a_{j_n})\times P(S_1^{(m)} = b_{j_n} ) \times \textrm{[other  terms]}
  \] in $W_j^{(n)}$, there is a corresponding term
  \[
    P(S_j^{(m)} = b_{j_n})\times P(S_1^{(m)} = a_{j_n} ) \times \textrm{[other  terms]}
  \] in $W_1^{(n)}$.
  Since by Lemma~\ref{lem3}
  \[
    \frac {P(S_j^{(m)} = a_{j_n})\times P(S_1^{(m)} = b_{j_n} )}{  P(S_j^{(m)} = b_{j_n})\times P(S_1^{(m)} = a_{j_n} )}
  \] converges uniformly to 1, we conclude
  \[
    \frac{W_{j_{n}}^{(n)}}{W_{1}} \rightarrow 1.
  \]
   We conclude that for any $\zeta>0$, we can write (for large enough $n$)
  \[
      (1-\zeta)<\frac{W_{j_{n}}^{(n)}}{W_{1}} < (1+\zeta),
  \]

  \[
    \sum_{j \in  \Pi(J^{(n)})} (1-\zeta)W_{1}^{(n)} < \sum_{j \in  \Pi(J^{(n)})} W_{j_{n}}^{(n)} < \sum_{j \in  \Pi(J^{(n)})} (1+\zeta)W_{1}^{(n)}
  \]
  and since
  $ \sum_{j \in  \Pi(J^{(n)})} W_{j_{n}}^{(n)} = 1,
  $
  $
    |\Pi(J^{(n)})| = N^{(n)},
  $ we have
  \[
  (1-\zeta)N^{(n)}W_{1}^{(n)} < 1 < (1+\zeta)N^{(n)}W_{1}^{(n)}
  \]
  so, we conclude that $ N^{(n)} W_1^{(n)} \rightarrow 1$ as $n \rightarrow \infty$. We can repeat the same argument for all users in set $j \in J^{(n)}$ and we get $N^{(n)} W_{j}^{(n)} \rightarrow 1$ as $n \rightarrow \infty$.
\end{proof}

Now to finish the proof of Theorem 1,
\begin{align}
   \no P( X_1(k)|\textbf{Y}^{(m)} , \Pi(J^{(n)}) ) =\ \ \ \ \ \ \ \ \ \ \
   \\ \no \sum_{j \in \Pi(J^{(n)}) } P(X_1(k)|\textbf{Y}^{(m)} , \Pi(1) = j , \Pi(J^{(n)})) \times \\ \no  P( \Pi(1) = j | \textbf{Y}^{(m)} , \Pi(J^{(n)}))
  =
  \\ \no \sum_{j \in \Pi(J^{(n)})} 1_{[Y^{(m)}_j(k) =1]}W_j^{(n)}  \triangleq Z_n.
\end{align}
But, since $Y^{(m)}_j(k) \sim Bernoulli(p_j^{(n)})$ and $p_j^{(n)} \rightarrow p_1$ for all $j \in \Pi(J^{(n)})$, by the law of large numbers we have: \[
  \frac{1}{N^{(n)}} \sum_{j \in \Pi(J^{(n)})} 1_{[Y^{(m)}_j(k) =1]} \rightarrow p_1
\]
\[
Z_n = \frac{1}{N^{(n)}} \sum_{j \in \Pi(J^{(n)})}(1_{[Y^{(m)}_j(k) =1]} ) (N^{(n)}W_j^{(n)})
\]
which $
  (N^{(n)}W_j^{(n)}) \rightarrow 1.
$
Thus, $ Z_n \rightarrow p_1$.

In conclusion $ X_1(k)|\textbf{Y}^{(m)} , \Pi(J^{(n)}) \xrightarrow{d} Bernoulli(p_1)$ which means that
\[
  H\left(X_1(k)|\textbf{Y}^{(m)} , \Pi(J^{(n)})\right) \rightarrow H(X_1(k))
\]
\[
  \Rightarrow  H\left(X_1(k)|\textbf{Y}^{(m)}\right) \leq   H\left(X_1(k)|\textbf{Y}^{(m)} , \Pi(J^{(n)})\right) \rightarrow H(X_1(k))
\]
\[
  \Rightarrow I\left(X_1(k);\textbf{Y}^{(m)})\right) \rightarrow 0
\]

\section{Proof of Lemma 1}\label{sec:app_b}
Here, we provide a formal proof for Lemma \ref{lem-Markov-suff} which we restate as follows. In the Markov chain setting of section \ref{sec:MC}, we have the following: Given $\textrm{Perm} \big(\mathbf{M}^{(m)}, \Pi^{(n)} \big)$, the random matrix $\textbf{Y}^{(m)}$ and the random permutation $\Pi^{(n)}$ are conditionally independent. That is
 \begin{align}\label{eq:sufficient-mc}
 P\left(\Pi^{(n)}=\pi \ \ \bigg{ | } \ \ \textbf{Y}^{(m)}=\mathbf{y}, \textrm{Perm} \big(\mathbf{M}^{(m)}, \Pi^{(n)} \big)=\mathbf{m} \right)  =
 \\ \no P\left(\Pi^{(n)}=\pi \ \ \bigg{ | } \ \ \textrm{Perm} \big(\mathbf{M}^{(m)}, \Pi^{(n)} \big)=\mathbf{s} \right)
\end{align}

\begin{proof}

Remember
\begin{align}
\nonumber \textbf{Y}^{(m)} &=\textrm{Perm}\left( \textbf{X}_{1}^{(m)}, \textbf{X}_{2}^{(m)}, \cdots,  \textbf{X}_{n}^{(m)}; \Pi^{(n)} \right) \\
\nonumber &=\left( \textbf{X}_{\Pi^{-1}(1)}^{(m)}, \textbf{X}_{\Pi^{-1}(2)}^{(m)}, \cdots,  \textbf{X}_{\Pi^{-1}(n)}^{(m)}\right).
\end{align}
Note that $\textbf{Y}^{(m)}$ (and therefore $\mathbf{y}$) is an $m$ by $n$ matrix, so we can write
\begin{align}
\nonumber \mathbf{y} &=\left( \textbf{y}_{1}, \textbf{y}_{2}, \cdots,  \textbf{y}_{n} \right),
\end{align}
where for $u=1,2,\cdots, n$, we have
\[\mathbf{y}_u = \begin{bmatrix}
\nonumber       y_u(1) \\ y_u(2) \\ \vdots \\y_u(m) \end{bmatrix}.
\]
Also, $\mathbf{m}$ is a collection of $n$ matrices so we can write

\begin{align}
\nonumber \mathbf{m}&=\left( \mathbf{m}_{1}, \mathbf{m}_{2}, \cdots, \mathbf{m}_{n} \right).
\end{align}

For an $r \times r$ matrix $\textbf{m}=[m(i,j)]$, let's define $D(\mathbf{m})$ as the set of sequences $(x_1,x_2,\cdots,x_m) \in \{1,2,\cdots,r\}^{m}$ that satisfy the following properties:
\begin{enumerate}
  \item $x_0=1$;
  \item The number of transitions from $i$ to $j$ in $(x_1,x_2,\cdots,x_m)$ is equal to $m_{ij}$ for all $i$ and $j$. That is, the number of indices $k$ for which we have $x_k=i$ and $x_{k+1}=j$ is equal to $m(i,j)$.
\end{enumerate}

We now show that the two sides of Equation \ref{eq:sufficient-mc} are equal. The right hand side probability can be written as

 \begin{align}%\label{eq:suff}
 \no P\left(\Pi^{(n)}=\pi \ \ \bigg{ | } \ \ \textrm{Perm} \big(\mathbf{M}^{(m)}, \Pi^{(n)}\big)=\mathbf{m} \right)\ \ \ \ \ \ \ \ \ \ \ \ \ \ \ \
 \\ \no = \frac{P\left( \textrm{Perm} \big(\mathbf{M}^{(m)}, \Pi^{(n)}\big)=\mathbf{m} \ \ \bigg{ | } \ \ \no \Pi^{(n)}=\pi  \right) P\left( \Pi^{(n)}=\pi  \right)}{P \bigg(\textrm{Perm} \big(\mathbf{M}^{(m)}, \Pi^{(n)}\big)=\mathbf{m}\bigg)}\\
\no =\frac{P\left( \textrm{Perm} \big(\mathbf{M}^{(m)}, \pi \big)=\mathbf{m} \ \ \bigg{ | } \ \ \no \Pi^{(n)}=\pi  \right)} { n! P \bigg(\textrm{Perm} \big(\mathbf{M}^{(m)},\Pi^{(n)}\big)=\mathbf{m} \bigg)}\\
\no =\frac{P\left( \textrm{Perm} \big(\mathbf{M}^{(m)}, \pi \big)=\mathbf{m} \right)} { n! P \bigg(\textrm{Perm} \big(\mathbf{M}^{(m)},\Pi^{(n)}\big)=\mathbf{m} \bigg)}.
\end{align}
Now note that
 \begin{align}%\label{eq:suff}
 \no P\left( \textrm{Perm} \big(\mathbf{M}^{(m)}, \pi \big)=\mathbf{m} \right) =P \left( \bigcap_{j=1}^{n} \left( \textbf{M}_{\pi^{-1}(j)}^{(m)}=\textbf{m}_j\right) \right)\\
 \no =P \left( \bigcap_{u=1}^{n} \left( \textbf{M}_{u}^{(m)}=\textbf{m}_{\pi (u)}\right) \right)\\
 \no =\prod_{u=1}^{n} P\left( \textbf{M}_{u}^{(m)}=\textbf{m}_{\pi (u)}\right)\end{align}
 \begin{align}
 \no = \prod_{u=1}^{n}  \sum_{\substack{(x_1,x_2,\cdots,x_m) \in \\ D(\textbf{m}_{\pi (u)})}}\hspace{-2mm} P( X_{u}(1)=x_1, X_{u}(2)=x_2, \cdots, X_{u}(m)=x_m)\\
 \no = \prod_{u=1}^{n}  \sum_{(x_1,x_2,\cdots,x_m) \in D(\textbf{m}_{\pi (u)})}  \prod_{i,j} p_u(i,j)^{\textbf{m}_{\pi (u)}(i,j)} \\
 \no = \prod_{u=1}^{n} \left( |D(\textbf{m}_{\pi (u)})|  \prod_{i,j} p_u(i,j)^{\textbf{m}_{\pi (u)}(i,j)} \right)\\
 \no = \left(  \prod_{k=1}^{n} |D(\textbf{m}_k)|\right) \left(  \prod_{u=1}^{n}  \prod_{i,j} p_u(i,j)^{\textbf{m}_{\pi (u)}(i,j)}\right)
\end{align}
Similarly, we obtain
 \begin{align}%\label{eq:suff}
 \no P \bigg(\textrm{Perm} \big(\mathbf{M}^{(m)},\Pi^{(n)}\big)=\mathbf{m} \bigg) = \ \ \ \ \ \ \ \ \ \ \ \ \ \ \ \ \ \ \ \
 \\ \no \sum_{\textrm{all permutations }\pi'} \hspace{-6mm}P\left( \textrm{Perm} \big(\mathbf{M}^{(m)}, \pi' \big)=\mathbf{m} \ \ \bigg{ | } \ \ \no \Pi^{(n)}=\pi'  \right) P\left( \Pi^{(n)}=\pi'  \right) \\
 \no = \frac{1}{n!} \sum_{\textrm{all permutations }\pi'}  \left(  \prod_{k=1}^{n} |D(\textbf{m}_k)|\right) \left(  \prod_{u=1}^{n}  \prod_{i,j} p_u(i,j)^{\textbf{m}_{\pi' (u)}(i,j)}\right)\\
 \no = \frac{1}{n!} \left(  \prod_{k=1}^{n} |D(\textbf{m}_k)|\right)  \sum_{\textrm{all permutations }\pi'}   \left(  \prod_{u=1}^{n}  \prod_{i,j} p_u(i,j)^{\textbf{m}_{\pi' (u)}(i,j)}\right).
\end{align}
Thus, we conclude that the right hand side of Equation \ref{eq:sufficient-mc} is equal to
 \begin{align}%\label{eq:suff}
 \no  \frac{\prod_{u=1}^{n}  \prod_{i,j} p_u(i,j)^{\textbf{m}_{\pi (u)}(i,j)}}{\sum_{\textrm{all permutations }\pi'}   \left(  \prod_{u=1}^{n}  \prod_{i,j} p_u(i,j)^{\textbf{m}_{\pi' (u)}(i,j)}\right) }.
\end{align}

Now let's look at the left hand side of Equation \ref{eq:sufficient-mc}. We can write
\begin{align}
 \no P\left(\Pi^{(n)}=\pi \ \ \bigg{ | } \ \ \textbf{Y}^{(m)}=\mathbf{y}, \textrm{Perm} \big(\mathbf{M}^{(m)}, \Pi^{(n)} \big)=\mathbf{m} \right)  =
 \\ \no  P\left(\Pi^{(n)}=\pi \ \ \bigg{ | } \ \ \textbf{Y}^{(m)}=\mathbf{y} \right).
\end{align}
This is because $\textrm{Perm} \big(\mathbf{M}^{(m)}, \Pi^{(n)} \big)$ is a function of $\textbf{Y}^{(m)}$. We have

 \begin{align}
 \no P\left(\Pi^{(n)}=\pi \ \ \bigg{ | } \ \ \textbf{Y}^{(m)}=\mathbf{y} \right)=  \ \ \ \ \ \ \ \ \ \ \ \ \
 \\ \no \frac{ P\left(\textbf{Y}^{(m)}=\mathbf{y}\ \ \bigg{ | } \ \  \Pi^{(n)}=\pi   \right) P\left(  \Pi^{(n)}=\pi \right)}{P\left(\textbf{Y}^{(m)}=\mathbf{y} \right)}
\end{align}

We have
 \begin{align}
 \no P\left(\textbf{Y}^{(m)}=\mathbf{y}\ \ \bigg{ | } \ \  \Pi^{(n)}=\pi   \right)= \ \ \ \ \ \ \ \ \ \ \ \ \ \ \ \ \ \ \
 \\ \no \prod_{u=1}^{n} P\bigg{(} X_{u}(1)=y_{\pi (u)}(1), X_{u}(2)=y_{\pi (u)}(2), \\ \no \cdots, X_{u}(m)=y_{\pi (u)}(m)\bigg{)}\\
 \no = \prod_{u=1}^{n}  \prod_{i,j} p_u(i,j)^{\textbf{m}_{\pi (u)}(i,j)}.
\end{align}
Similarly, we obtain
 \begin{align}
 \no P\left(\textbf{Y}^{(m)}=\mathbf{y} \right)&= \frac{1}{n!} \sum_{\textrm{all permutations }\pi'} \prod_{u=1}^{n} \prod_{i,j} p_u(i,j)^{\textbf{m}_{\pi' (u)}(i,j)}
\end{align}
Thus, we conclude that the left hand side of Equation \ref{eq:sufficient-mc} is equal to
 \begin{align}%\label{eq:suff}
 \no  \frac{\prod_{u=1}^{n}  \prod_{i,j} p_u(i,j)^{\textbf{m}_{\pi (u)}(i,j)}}{\sum_{\textrm{all permutations }\pi'}   \left(  \prod_{u=1}^{n}  \prod_{i,j} p_u(i,j)^{\textbf{m}_{\pi' (u)}(i,j)}\right) },
\end{align}
which completes the proof.

\end{proof}

\bibliographystyle{IEEEtran}
\bibliography{journal}

%\begin{IEEEbiography}{Zarrin Montazeri}
%Biography text here.
%\end{IEEEbiography}
%
%
%
%\begin{IEEEbiography}{Amir Houmansadr}
%Biography text here.
%\end{IEEEbiography}
%
%\begin{IEEEbiography}{Hossein Pishro-Nik}
%Biography text here.
%\end{IEEEbiography}

% that's all folks
\end{document}